\documentclass[11pt, letterpaper, onecolumn, accepted=2023-10-25]{quantumarticle}
\pdfoutput=1

\usepackage{amsmath}
\usepackage{amsthm}
\usepackage{amssymb}
\usepackage{algorithm}
\usepackage{subfigure}
\usepackage[dvipsnames]{xcolor}
\usepackage{graphicx}
\usepackage{algpseudocode}
\usepackage{comment}
\usepackage{mathtools}
\usepackage{enumitem}
\usepackage{tikz}
\usepackage{url}
\usepackage{float}
\usepackage[margin=1in]{geometry}
\usepackage{qcircuit}
\usepackage{bbm}
\usepackage[noblocks]{authblk}

\usepackage{tabularx}

\usepackage{fancyref}
\usepackage{hyperref}
\hypersetup{colorlinks=true,citecolor=red,linkcolor=blue}

\usetikzlibrary{positioning}
\usepackage{pgfplots}
	\usetikzlibrary{
		calc,
		patterns,
		positioning
	}
	\pgfplotsset{
		compat=1.16,
		samples=200,
		clip=false,
		my axis style/.style={
			axis x line=middle,
			axis y line=middle,
			legend pos=outer north east,
			axis line style={
				->,
			},
			legend style={
				font=\footnotesize
			},
			label style={
				font=\footnotesize
			},
			tick label style={
				font=\footnotesize
			},
			xlabel style={
				at={
					(ticklabel* cs:1)
				},
				anchor=west,
				font=\footnotesize,
			},
			ylabel style={
				at={
					(ticklabel* cs:1)
				},
				anchor=west,
				font=\footnotesize,
			},
			xlabel=$T_{tot}$,
			ylabel=$T_{max}$
		},
	}
	\tikzset{
		>=stealth
	}

\newtheorem{theorem}{Theorem}[section]
\newtheorem{lemma}[theorem]{Lemma}

\newtheorem{corollary}[theorem]{Corollary}

\newtheorem{fact}[theorem]{Fact}
\newtheorem{remark}[theorem]{Remark}
\newtheorem{claim}[theorem]{Claim}

\newcommand{\vertiii}[1]{{\left\vert\kern-0.25ex\left\vert\kern-0.25ex\left\vert #1 \right\vert\kern-0.25ex\right\vert\kern-0.25ex\right\vert}}

\newcommand{\lrb}[1]{\left ( #1 \right )}
\newcommand{\lrsb}[1]{\left [ #1 \right ]}
\newcommand{\lrcb}[1]{\left \{ #1 \right \}}

\newcommand{\abs}[1]{\left | #1 \right |}
\newcommand{\norm}[1]{\left \| #1 \right \|}

\newcommand{\lonenorm}[1]{\left \| #1 \right \|_{1}}

\newcommand{\mytr}[1]{\operatorname{tr}\lrsb{#1}}

\newcommand{\ket}[1]{\left| #1 \right\rangle}
\newcommand{\bra}[1]{\left\langle #1 \right|}
\newcommand{\ketbra}[2]{|#1\rangle \langle #2|}

\renewcommand{\exp}[1]{\operatorname{exp} \lrb{#1}}

\renewcommand{\log}[1]{\operatorname{log} \lrb{#1}}

\renewcommand{\ln}[1]{\operatorname{ln} \lrb{#1}}

\newcommand{\R}{\mathbb{R}}
\newcommand{\C}{\mathbb{C}}
\newcommand{\Z}{\mathbb{Z}}

\renewcommand{\P}[1]{\mathbb{P}\lrsb{#1}}
\newcommand{\E}[1]{\mathbb{E}\lrsb{#1}}
\newcommand{\wt}[1]{\widetilde{{\mathcal #1}}}

\newcommand{\defeq}{\coloneqq}

\DeclareMathOperator{\rep}{Re}
\DeclareMathOperator{\imp}{Im}
\renewcommand{\i}{\mathbf{i}}

\newcommand{\bfX}{{\bf X}}
\newcommand{\bfY}{{\bf Y}}
\newcommand{\bfZ}{{\bf Z}}
\newcommand{\bft}{{\bf t}}
\newcommand{\bfb}{{\bf b}}

\newcommand*{\fancyrefthmlabelprefix}{thm}
\newcommand*{\fancyreflemlabelprefix}{lem}
\newcommand*{\fancyrefcorlabelprefix}{cor}
\newcommand*{\fancyrefdefilabelprefix}{defi}
\frefformat{plain}{\fancyreflemlabelprefix}{lemma\fancyrefdefaultspacing#1}
\Frefformat{plain}{\fancyreflemlabelprefix}{Lemma\fancyrefdefaultspacing#1}
\frefformat{plain}{\fancyrefthmlabelprefix}{theorem\fancyrefdefaultspacing#1}
\Frefformat{plain}{\fancyrefthmlabelprefix}{Theorem\fancyrefdefaultspacing#1}
\frefformat{plain}{\fancyrefcorlabelprefix}{corollary\fancyrefdefaultspacing#1}
\Frefformat{plain}{\fancyrefcorlabelprefix}{Corollary\fancyrefdefaultspacing#1}
\frefformat{plain}{\fancyrefdefilabelprefix}{definition\fancyrefdefaultspacing#1}
\Frefformat{plain}{\fancyrefdefilabelprefix}{Definition\fancyrefdefaultspacing#1}

\title{Quantum algorithm for ground state energy estimation using circuit depth with exponentially improved dependence on precision}

\author[1]{Guoming Wang}
\author[2]{Daniel Stilck Fran\c{c}a}
\author[1,3]{Ruizhe Zhang}
\author[1,4]{Shuchen Zhu}
\author[5]{Peter D. Johnson}

\affil[1]{Zapata Computing Canada Inc., Toronto, ON M5V 2Y1, Canada}
\affil[2]{Univ Lyon, ENS Lyon, UCBL, CNRS, Inria, LIP, F-69342, Lyon Cedex 07, France}
\affil[3]{Department of Computer Science, University of Texas at Austin, Austin, TX 78712, USA}
\affil[4]{Department of Computer Science, Georgetown University, Washington, DC 20057, USA}
\affil[5]{Zapata Computing Inc., Boston, MA 02110 USA}

\begin{document}

\maketitle

\begin{abstract}
    A milestone in the field of quantum computing will be solving problems in quantum chemistry and materials faster than state-of-the-art classical methods.
    The current understanding is that achieving quantum advantage in this area will require some degree of fault tolerance.
    While hardware is improving towards this milestone, optimizing quantum algorithms also brings it closer to the present.
    Existing methods for ground state energy estimation are costly in that they require a number of gates per circuit that grows exponentially with the desired number of bits in precision. We reduce this cost exponentially, by developing a ground state energy estimation algorithm for which this cost grows linearly in the number of bits of precision.
    Relative to recent resource estimates of ground state energy estimation for the industrially-relevant molecules of ethylene-carbonate and PF$_6^-$, the estimated gate count and circuit depth is reduced by a factor of 43 and 78, respectively.
    Furthermore, the algorithm can use additional circuit depth to reduce the total runtime. 
    These features make our algorithm a promising candidate for realizing quantum advantage in the era of early fault-tolerant quantum computing.
\end{abstract}

\section{Introduction}

When will quantum computers solve valuable problems that are out of reach for state-of-the-art classical approaches?
To understand this future moment of quantum advantage, we must identify what computational problems are most apt and determine what quantum algorithms will be able solve them in the nearest time frame.
Despite some recent challenges being illuminated \cite{larsson2022chromium}, estimating the ground state energy of quantum systems \cite{aspuru2005simulated} remains one of the leading contenders for the first realization of practial quantum advantage.
Solving this problem efficiently with a quantum computer would be of high value to areas including combustion \cite{gonthier2020identifying}, batteries \cite{kim2022fault, delgado2022simulate}, and catalysts \cite{goings2022reliably}.
Considering that the progress in quantum hardware has consistently improved, we are urged to investigate: what are the minimal quantum resources needed to realize quantum advantage with ground state energy estimation {(GSEE)}?

\vspace{0.5cm}

\noindent\textbf{\large{Previous methods for ground state energy estimation}}

\noindent {To estimate ground state energy on early quantum computers, researchers have designed more efficient algorithms using fewer qubits, gates, and ancillas.} Most of these algorithms are based on the variational quantum eigensolver (VQE) \cite{peruzzo2014variational}. 
These algorithms do not have performance guarantees and recent works have identified roadblocks for the practicality of such approaches through the measurement problem \cite{gonthier2020identifying, johnson2022reducing} and issues with optimization \cite{mcclean2018barren, anschuetz2022beyond}.
Yet, quantum algorithms with performance guarantees \cite{aspuru2005simulated, babbush2018encoding, lin2022heisenberg, dong2022ground} demand unfortunately-large quantum resources, requiring hundreds of logical qubits and over billions of operations per circuit (e.g. greater than $10^{10}$ T gates \cite{kim2022fault}).
The error correction overhead needed to run such quantum circuits is far beyond what can be realized on today's hardware.
Towards realizing practical quantum advantage sooner, we develop quantum algorithms in the goldilocks regime of having provable performance guarantees while also exponentially reducing the required number of operations.

The operations involved in many GSEE algorithms, including ours, are controlled time evolution operations, $c$-$\exp{\i Ht}$. 
For the algorithms considered, the total number of operations per circuit and the circuit depth are proportional (ignoring logarithmic factors) to the evolution time $t$.
We will refer to this measure as the \emph{circuit depth}.
The circuit depth required by a GSEE algorithm is typically costed in terms of $\epsilon$, the target accuracy of the ground state energy estimate, and $\eta$, a lower bound on the overlap of the input state $\rho$ with the ground state\footnote{An important caveat for all known ground state energy estimation methods (c.f Table I of \cite{dong2022ground}) is that if $\eta$ is extremely small, then we have little hope of accurately estimating the ground state energy. Thus, it is common to assume that the Hamiltonian of interest admits a good ground state approximation.
}.
In contrast to previous GSEE methods, the costs of our GSEE algorithms will depend on $\Delta$, a lower bound on the energy gap (i.e. the difference between the smallest and next-smallest eigenvalue of $H$), typically governs the performance of ground state \emph{preparation} methods \cite{dong2022ground}. 
With these parameters established we are able to describe the costs of previous methods for GSEE that aim to minimize the quantum resources.
{Recently, building upon the ideas of \cite{somma2019quantum}, Lin and Tong \cite{lin2022heisenberg} designed a quantum algorithm for ground state energy estimation with circuit depths scaling as $\tilde{O}(1/\epsilon)$, which requires just a single ancilla qubit and involves no costly circuit operations beyond controlled time evolutions. 
This algorithm requires running the circuits of depth $\tilde{O}(\epsilon^{-1})$ multiple times, leading to a total runtime of $\tilde{O}(\epsilon^{-1}\eta^{-2})$.
Later, Dong et al.~\cite{dong2022ground} improved on this result, developing an algorithm with similar characteristics, yet achieving a runtime of $\tilde{O}(\epsilon^{-1}\eta^{-1})$. Both of these methods achieve the so-called Heisenberg-limited scaling in the runtime with respect to $\epsilon$. }
{Meanwhile, Wan et al.} \cite{wan2021randomized} combined the ideas in \cite{lin2022heisenberg} with the principles of randomized Hamiltonian evolution (QDRIFT) \cite{campbell2019random} to propose a ground state energy estimation algorithm which trades off between number of non-Clifford gates and runtime. {Ding and Lin \cite{ding2023even} and Ni et al. \cite{ni2023onlowdepth} also discovered low-depth algorithms for  phase estimation which achieve the Heisenberg-limited scaling of quantum runtime while having a diminishing prefactor for the maximum circuit depth when the initial state approaches the target eigenstate.} 
{There has also been progress in developing low-cost (or low-depth) quantum algorithms for simultaneous estimation of multiple eigenvalues in the past years \cite{obrien2019quantum, somma2019quantum, dutkiewicz2022heisenberg, ding2023simultaneous, li2023onadative}.
Our work is mostly influenced by \cite{somma2019quantum, lin2022heisenberg} which promoted the idea of using quantum circuits to general time signals and utilizing the tools from classical signal processing to analyze them. This time series analysis reveals the target property of the spectrum of the Hamiltonian.}

{All of the GSEE methods prior to ours employ quantum circuits whose depth scales inversely with $\epsilon$ the target accuracy \footnote{{After the first version of this paper was posted on arXiv, Ding and Lin \cite{ding2023even} proposed a GSEE algorithm which achieves essentially the same circuit depth and quantum runtime as ours. Although both their algorithm and our algorithm are based on classical postprocessing of the data from Hadamard tests, the specific ways of exploiting such data are quite different. The algorithm of \cite{ding2023even} takes a direct data fitting approach and solves nonlinear optimization} {problems to find the values that best match the observed data. By contrast, we take an indirect Fourier filtering approach and use the data to construct the convolution of the spectral measure and a Gaussian derivative filter, and then infer the ground state energy from the shape of this convolution. The classical computation is simpler here. On the other hand, \cite{ding2023even} also proposed techniques for reducing the circuit depth for GSEE when the overlap between the initial state and the ground state is large (without a spectral-gap assumption). This scenario is not considered in this work.}}. In the early fault-tolerant setting, this circuit depth  requirement may} place a limit on the size of problem instances that such methods can accommodate. It is reasonable to assume that early fault-tolerant quantum computers will be limited in number of physical qubits. Larger problem instances require more logical qubits to be encoded in these physical qubits. This entails an increase in the error rate for logical operations (as the code distance has been reduced). This increase in error rate compromises the number of operations that can be implemented per circuit (assuming a fixed circuit error tolerance of the algorithm). To ``unlock’’ larger instances requires running a GSEE algorithm that can succeed with fewer operations per circuit (or less circuit depth). Motivated by these considerations, the question addressed by this work is: \emph{is it possible to unlock larger problem instances for early fault-tolerant quantum computers through an exponential improvement in the accuracy-dependence scaling of circuit depth in ground state energy estimation?}

\vspace{0.5cm}

\noindent\textbf{\large{Summary of main results}}

\noindent We develop and analyze low-depth ground state energy estimation (GSEE) algorithms with high accuracy for which the circuit depth scales exponentially better than $\tilde{O}(1/\epsilon)$. As is typical, the circuit depth and quantum runtime of the algorithm is measured in terms of the number of controlled evolution operations c-$\exp{-2\pi \i H}$ referred to as \emph{Hamiltonian evolution time}.
\begin{theorem}[Low-depth GSEE, informal version of Theorem~\ref{thm:gsee_main_formal}]\label{thm:gsss_main_informal}
Let $H$ be a Hamiltonian with spectral gap at least $\Delta$. Suppose we can prepare an initial state $\rho$ such that the overlap with the ground state satisfies $\langle E_0|\rho|E_0\rangle \geq \eta$.
Given $\Delta$, $\eta$, and  sufficiently small $\epsilon$, there exists an algorithm to estimate the ground state energy within accuracy $\epsilon$ with high probability such that:
\begin{itemize}
    \item The circuit depth, measured in maximal Hamiltonian evolution time,
    is 
    \begin{align}
        {\cal T}_{\rm max}={\cal O}\lrb{\Delta^{-1}\cdot \mathrm{poly}\log{\epsilon^{-1}\eta^{-1}\Delta}}.
    \end{align}
    \item The quantum runtime, measured in total Hamiltonian evolution time, is  
    \begin{align}
        {\cal T}_{\rm tot}={\cal O}\lrb{\eta^{-2}\epsilon^{-2}\Delta\cdot \mathrm{poly}\log{\epsilon^{-1}\eta^{-1}\Delta}}.
    \end{align}
\end{itemize}
\label{thm:gsee_main_informal}
\end{theorem}

\definecolor{mygray}{gray}{0.6}

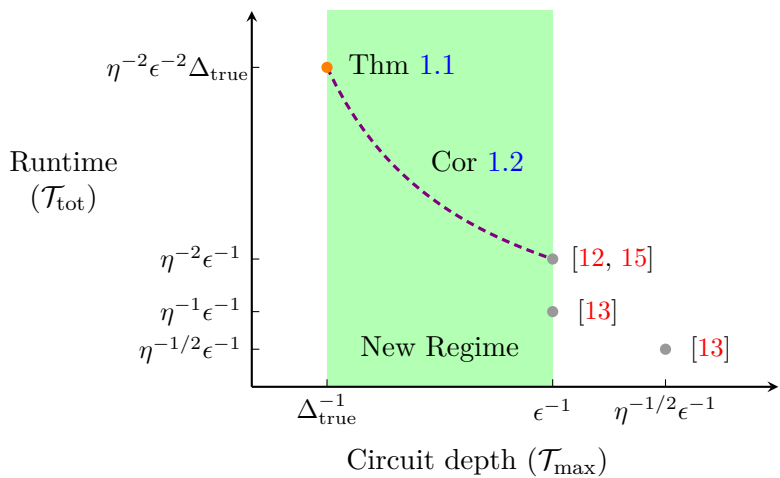
\begin{figure}
    \centering

\begin{tikzpicture}
\node[rectangle,
    fill = green!30!white,
    minimum width = 3cm, 
    minimum height = 5cm,
    anchor = south west] (r) at (1,0) {};
\node[] at (2.5, 0.5) {New Regime};
\draw[thick=0.1, ->] (0, 0) -- (7, 0) node[below=1mm] {};
\node[] at (3, -1) {Circuit depth (${\cal T}_{\max}$)};
\draw[thick=0.1,->] (0, 0) -- (0, 5) node[left] {};
\node[] at (-2.5, 3.0) {Runtime};
\node[] at (-2.5, 2.5) {(${\cal T}_{\rm tot}$)};
\draw (1, 0.1) -- (1, 0);
\draw (0.1, 1.7) -- (0, 1.7) node[left]{\small $\eta^{-2}\epsilon^{-1}$};
\draw (0.1, 1) -- (0, 1) node[left]{\small $\eta^{-1}\epsilon^{-1}$};
\draw (0.1, 0.5) -- (0, 0.5) node[left]{\small $\eta^{-1/2}\epsilon^{-1}$};
\node[] at (1, -0.3) {\small $\Delta_{\textup{true}}^{-1}$};
\draw (4, 0) -- (4, 0.1);
\draw (0, 4.25) -- (0.1, 4.25) node[left] {\small $\eta^{-2}\epsilon^{-2}\Delta_{\textup{true}}$};
\node[] at (4, -0.3) {\small $\epsilon^{-1}$};
\draw (5.5, 0) -- (5.5, 0.1);
\node[] at (5.5, -0.3) {\small $\eta^{-1/2}\epsilon^{-1}$};
\draw[violet, densely dashed, very thick, domain=1:4] plot (\x,{8.5/(\x+1)});
\node[] at (3, 3) {Cor~\ref{cor:interpolation_informal}};

\node at (1,4.25) [circle,fill=orange,inner sep=1.5pt, label={[shift={(1,-0.3)}]Thm~\ref{thm:gsee_main_informal}}]{};
\node at (4,1.7) [circle,fill=mygray,inner sep=1.5pt, label={[shift={(0.8,-0.4)}]\small \cite{lin2022heisenberg,wan2021randomized}}]{};
\node at (4,1) [circle,fill=mygray,inner sep=1.5pt, label={[shift={(0.6,-0.4)}]\small \cite{dong2022ground}}]{};
\node at (5.5,0.5) [circle,fill=mygray,inner sep=1.5pt, label={[shift={(0.6,-0.4)}]\small \cite{dong2022ground}}]{};
\end{tikzpicture}
    \caption{
    This figure shows the landscape of early fault-tolerant GSEE algorithms plotted according to their runtime and circuit depth measured in terms of total evolution time (${\cal T}_{\rm tot}$) and maximal evolution time (${\cal T}_{\max}$), respectively.
    The green region indicates the new low-depth regime introduced in this work.
    The orange dot corresponds to the $\Delta^{-1}$-depth GSEE algorithm (Theorem~\ref{thm:gsee_main_informal}) when $\Delta = \Delta_{\textup{true}}^{-1}$, and the curve shows the smooth trade off between ${\cal T}_{\max}$ and ${\cal T}_{\rm tot}$ described in Corollary \ref{cor:interpolation_informal}. We also remark that the right-most dot which shows an algorithm in \cite{dong2022ground} requires multiple ancilla qubits and multi-qubit controlled operations, whereas the algorithms in this work and \cite{lin2022heisenberg,wan2021randomized} only use a single ancilla qubit. For simplicity, we have ignored all the poly-logarithmic factors.} 
    \label{fig:intro}
\end{figure}

\begin{table}[]
\begin{tabular}{ccccc}
\hline
\multicolumn{1}{|c|}{Molecule} & \multicolumn{1}{c|}{Gap} & \multicolumn{1}{c|}{Gap Lower Bound } & \multicolumn{1}{c|}{Gate Reduction \cite{kim2022fault}} & \multicolumn{1}{c|}{Gate Reduction \cite{lin2022heisenberg}} \\ \hline
\multicolumn{1}{|c|}{EC} & \multicolumn{1}{c|}{$264\pm 20$ mHa} & \multicolumn{1}{c|}{244 mHa} & \multicolumn{1}{c|}{$43\times$} & \multicolumn{1}{c|}{$16\times$} \\ \hline
\multicolumn{1}{|c|}{$\textup{PF}_6^6$} & \multicolumn{1}{c|}{$468\pm20$ mHa} & \multicolumn{1}{c|}{$448$ mHa} & \multicolumn{1}{c|}{$78\times$} & \multicolumn{1}{c|}{$28\times$} \\ \hline
&    &    &    &                      
\end{tabular}
\caption{
This table displays estimated circuit cost savings afforded by Algorithm \ref{alg:low_depth_gsee} for two molecules relevant to battery design analyzed in previous work \cite{kim2022fault}.
For these molecules in the 
cc-pVDZ basis,
we can estimate the energy gaps using EOM-CCSD calculations with ORCA software \cite{neese2012orca,neese2018software}.
The target accuracy considered in \cite{kim2022fault} was $\epsilon = 1$ mHa.
The standard approach to quantum phase estimation (ignoring the cost due to imperfect ground state preparation) uses a circuit with $2/\epsilon$ applications of $c$-$\exp{2\pi \i H}$ to achieve an $\epsilon$ accurate estimate with high probability.
We include the various logarithmic factors (c.f. Algorithm \ref{alg:low_depth_gsee}) {and set an input state overlap value of $\eta = 1/1000$, which is conservative (i.e. lower) relative to the values found in \cite{tubman2018postponing}.}
In the last column we give cost reductions relative to recent methods \cite{lin2022heisenberg}, which use $ 2/\pi\epsilon$ applications of $c$-$\exp{2\pi \i H}$.}
\label{tab:resource_estimates}
\end{table}
\noindent Using the costs established in Theorem \ref{thm:gsss_main_informal}, Table \ref{tab:resource_estimates} provides resource estimates that show the reduction in gates per circuit for molecules of industrial interest. 
The reduction in gates per circuit affords a reduction in the fault-tolerant overhead required to implement ground state energy estimation.
These reductions may help to bring such problem instances within reach of earlier fault-tolerant quantum architectures, potentially realizing quantum advantage sooner.

For some molecules, it might be the case that the runtime of this low depth algorithm is too high to outperform state-of-the-art classical methods for solving the same problem.
Our second main result (c.f. Corollary \ref{rem:interpolation}) is that we can trade circuit depth for total runtime reduction. This gives a means of speeding up the overall algorithm.
\begin{corollary}
\label{cor:interpolation_informal}
Let $H$ be a Hamiltonian with spectral gap $\Delta_{\rm true}$. Suppose we can prepare an initial state $\rho$ such that the overlap with the ground state satisfies $\langle E_0|\rho|E_0\rangle \geq \eta$. Then for arbitrary $\alpha \in [0, 1]$, given $\Delta_{\rm true}$, $\eta$ and sufficiently small $\epsilon$, there exists an algorithm to estimate the ground state energy within accuracy $\epsilon$ with high probability such that:
\begin{itemize}
    \item The circuit depth, measured in maximal Hamiltonian evolution time, 
    is 
    \begin{align}
        {\cal T}_{\rm max}=\tilde{{\cal O}}\lrb{\epsilon^{-\alpha} \Delta_{\textup{true}}^{-1+\alpha}}.
    \end{align}
    \item The quantum runtime, measured in total Hamiltonian evolution time, is 
    \begin{align}
        {\cal T}_{\rm tot}=\tilde{{\cal O}}\lrb{\eta^{-2}\epsilon^{-2+\alpha}\Delta_{\textup{true}}^{1-\alpha}}.
    \end{align}
\end{itemize}
\end{corollary}
\noindent Through the era of early fault-tolerant quantum computing, as quantum architectures are able to realize deeper quantum circuits, the trade-off in Corollary \ref{cor:interpolation_informal} may lead to a crossover point into quantum advantage.


\section{Low-depth ground state energy estimation}

\begin{figure}[t]
    \centering
    \includegraphics[width=0.95\textwidth,trim={0.05cm 0 0.05cm 0},clip]{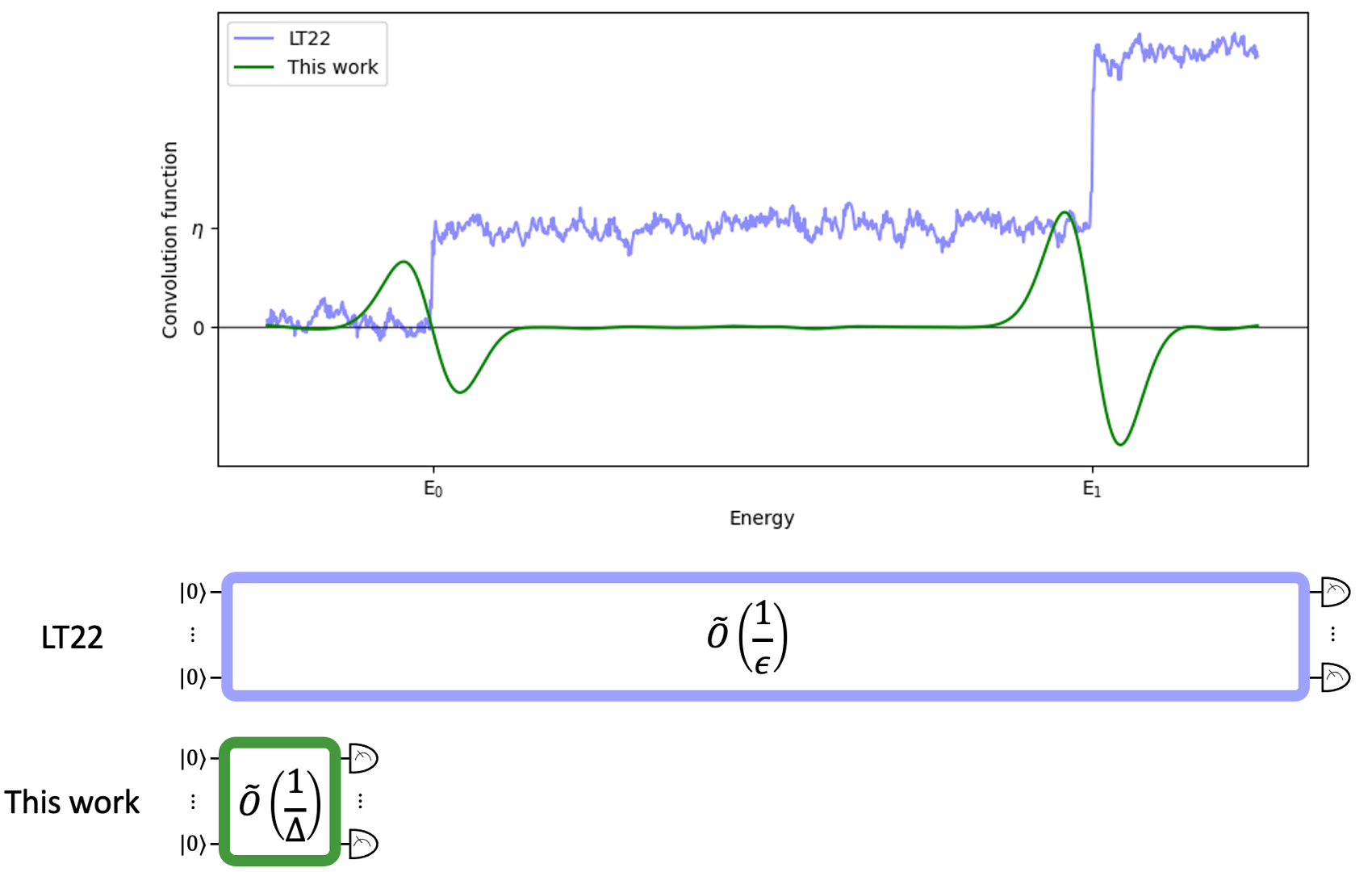}
    \caption{
    {This figure qualitatively compares the convolution functions and circuit depths used in the ground state energy estimation method of Ling and Tong \cite{lin2022heisenberg} (blue curve) and the method developed here (green curve). 
    The blue curve is an example estimate of the approximate cumulative distribution function (see Eq.~(16) in \cite{lin2022heisenberg}), while the green curve is an example estimate of Eq.~\eqref{eqn:convolution} as output from Step \ref{step:conv_est} in the algorithm below.}
    The LT22 method uses a Heaviside convolution, while our method uses a Gaussian derivative convolution. Their method requires a steep jump in the convolution function, which necessitates $\tilde{O}(1/\epsilon)$-depth circuits. Our method only requires that the contribution of the excited state energies to the convolution function does not interfere too much with that of the ground state energy.
    This affords the use of a less-steep convolution function, which only requires $\tilde{O}(1/\Delta)$-depth circuits. The trade-off is that our method requires more samples, leading to an increased total runtime.}
    \label{fig:GSEE_circuits}
\end{figure}

Before introducing the method, we define the ground state energy estimation problem.
Suppose we are given a classical description of a quantum Hamiltonian $H$. This Hamiltonian has (unknown) spectral decomposition $H=\sum_{j=0}^{N-1} E_j \ket{E_j} \bra{E_j}$, where $E_0<E_1 \le E_2 \le ...\le E_{N-1}$ are the eigenvalues of $H$, and the $\ket{E_j}$'s are orthonormal eigenstates of $H$. Let $\rho$ be an easy-to-prepare state (of the same dimension as $H$). Let $p_j \defeq \bra{E_j} \rho \ket{E_j}$ be the overlap between $\rho$ and $\ket{E_j}$, for $0\le j\le N-1$. We assume that two numbers $\eta \in (0, 1)$ and $\Delta>0$ are given such that $p_0 = \bra{E_0} \rho \ket{E_0} \ge \eta$ and $E_1 - E_0 \ge \Delta$. Our goal is to estimate $E_0$ with accuracy $\epsilon$ and confidence $1-\delta$, i.e. to output a sample from a random variable $\hat{E}_0$ such that the failure probability satisfies
\begin{align}
    \P{|\hat{E}_0-E_0|>\epsilon}<\delta,
\end{align}
for given small $\epsilon>0$ and $\delta \in (0, 1)$. Furthermore, we want to achieve this by using 
limited-depth quantum circuits and classical post-processing of the quantum measurement outcome data. 

{Our algorithm does not require the Hamiltonian $H$ to be normalized to work, i.e. it does not rely on an assumption like $\norm{H} \le 1$. The circuit depth and quantum runtime of our algorithm will be measured in the maximal and total evolution time of $H$, respectively. Note that the gate complexity of simulating $e^{\i H t}$ is proportional to $\norm{H}t$. Therefore, rescaling $H$, $\Delta$ and $\epsilon$ by the same constant simultaneously does not change the number of elementary gates in our algorithm (see Theorem \ref{thm:gsee_main_formal}).}

\begin{figure}[ht]
    \centering
		\begin{displaymath}
		\Qcircuit @C=1.0em @R=1.2em {
			& & & &\\			
			\lstick{\ket{0}}
			&\gate{\mathrm{H}}	 &\ctrl{1}	& \gate{\mathrm{W}}
			& \gate{\mathrm{H}}			&\meter\\
			\lstick{\rho} 	 & \qw & \gate{e^{-\i H \tau}} 		 
			&\qw &\qw &\qw
		}		
		\end{displaymath}		
\caption{Hadamard test circuit parameterized by the Hamiltonian evolution time $\tau$. $\mathrm{H}$ is the Hadamard gate and $\mathrm{W}$ is either $I$ or $S^\dagger$, where $S$ is the phase gate.}
    \label{fig:hadamard_test}
\end{figure}
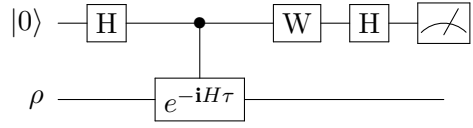

\paragraph{Time signals from Hadamard tests.}

In our GSEE algorithm, the role of the quantum computer is simply to provide statistical estimates of $\mytr{\rho e^{-\i H \tau}}$.
The quantum circuit we use to generate these estimates is known as a Hadamard test and is shown in Figure~\ref{fig:hadamard_test}. Labeling this outcome ${\bf b}\in \{+1, -1\}$, the average value of ${\bf b}$ output by the Hadamard test circuit is
\begin{align}
\label{eq:hadamard_estimator}
\mathbb{E}[{\bf b}] = \textup{Re}\left[\mytr{\rho e^{-\i H \tau}}\right],
\end{align}
when $W=I$ and it is equal to the imaginary part when $W=S^\dagger$ where $S$ is the phase gate.
It is helpful to view the quantity $\mytr{\rho e^{-\i H \tau}}$ as a complex-valued time signal, with $\tau$ being the time.
This time signal encodes information about the eigenvalues of $H$ and the density operator $\rho$.
In particular, if we can determine how this signal depends on the ground state energy $E_0$, then we might be able to estimate $E_0$ from the time signal.
Although we are unable to exactly determine the time signal, we can estimate the real and imaginary parts of the signal at any time $\tau$ to within any desired accuracy using sufficiently many Hadamard test measurement outcomes as described above.
The time cost of each Hadamard test is proportional to $\tau$ and the total time cost will depend on how many Hadamard tests, or samples, we take over the different chosen times $\tau$.

\paragraph{Filtering the spectrum}
Here we introduce the method for estimating and processing the signals from the Hadamard test data. The Fourier transform (or frequency signal) of the ideal time signal $\mytr{\rho e^{-\i H \tau}}$ 
is equal to the so-called spectral measure of $H$ associated with the initial state $\rho$ and is given by
\begin{align}\label{eq:def_p}
    p(x) \defeq \sum_{j=0}^{N-1} p_j \delta(x - E_j).
\end{align}
Although $p(x)$ itself cannot be determined exactly from Hadamard test data, we explain how to accurately estimate any convolution of $p(x)$ with a \emph{filter function} $f(x)$. 
For our purposes, the filter function is used as a tool for organizing the time signal data from a limited time window into useful information about the spectrum of $H$.
We briefly explain how to evaluate (or, rather, estimate) the complex number $(f \ast p)(x)$ for any given value $x$ using low-depth Hadamard test circuits.

Three key features make low-depth convolution estimation possible.
First, the convolution can be expressed as a linear combination of $\hat{p}(t)=\mytr{\rho e^{-2\pi\i H t}}$,
\begin{align}
(f\ast p)(x)=\int_{-\infty}^{\infty}\hat{f}(t)\hat{p}(t)e^{2\pi \i xt}dt,
\end{align}
where $\hat{f}(t)$ denotes the Fourier transform of $f(x)$.
Second, these traces can be estimated from the Hadamard test data as shown in 
Eq. \ref{eq:hadamard_estimator}.
Third, the circuit depth of each Hadamard test is proportional to $t$. We can limit the circuit depth used in the algorithm and still obtain an accurate estimate of the convolution by judiciously truncating the integral approximation
\begin{align}
\label{eq:truncation}
\int_{-\infty}^{\infty}\hat{f}(t)\hat{p}(t)e^{2\pi \i xt}dt \approx \int_{-T}^{T} \hat{f}(t) \mytr{\rho e^{-2\pi \i Ht}} e^{2\pi \i xt}dt.
\end{align}
As explained in detail in Algorithm \ref{alg:eval_conv}, the strategy we use to estimate $(f\ast p)(x)$ uses a so-called multi-level Monte Carlo approach.
An unbiased estimate of $(f\ast p)(x)$ is constructed by first (classically) sampling a time $t$  in $[-T,T]$ drawn from a distribution proportional to $|\hat{f}_T(t)|$.
Conditioned on this outcome $t$, a sample is then drawn from each of the real ($W=I$) and the imaginary ($W=S^{\dagger}$) part Hadamard tests with time $t$, returning $X$ and $Y$, respectively. From these outcome data $(t, X, Y)$, we construct a random variable whose average value is equal to $(f\ast p)(x)$,
\begin{align}
     Z(x;t,X,Y) = \|\hat{f}_{T}\|_{1}e^{2\pi \i (t x+\phi(t))}\cdot (X+\i Y).
     \label{eq:def_z}
\end{align}
where $e^{\i 2 \pi \phi(t)}=\hat{f}_T(t)/|\hat{f}_T(t)|$.
It is important to note that the samples can be generated ahead of time; the choice of where to evaluate $(f\ast p)(x)$ can be made after this data is gathered.

\paragraph{Designing the filters}
The filter function plays two roles in determining the algorithm performance.
First, the shape of the filter determines how easily the ground state energy can be determined from $(f\ast p)(x)$.
Second, the smoothness of the filter determines the severity of the truncation $T$ that can be withstood, and therefore the minimal viable circuit depth.
This second role is what affords the exponential reduction in circuit depth of our method.

Estimating the ground state energy from the convolution can be understood through an example.
In \cite{lin2022heisenberg} they choose $f(x)$ to be a periodic Heaviside function\footnote{$\Theta(x)=\begin{cases}1&\text{if}~x\in[2k\pi,(2k + 1)\pi)\\ 0 & \text{if}~x\in [(2k - 1)\pi,2k\pi)\end{cases}~~~\forall k\in \Z$.}.
As shown in Figure \ref{fig:GSEE_circuits}, this particular choice of convolution results in a series of steps, the first of which is located at the ground state energy $E_0$.
Their algorithm proceeds by using a binary search to locate this first step.
The drawback of this approach is that in order to resolve this first step to accuracy $\epsilon$, the truncation {order} must be $\tilde{\mathcal O}(1/\epsilon)$.
This means that the circuit depth of the Hadamard test scales as $\tilde{\mathcal O}(1/\epsilon)$.
We design a filter function and energy estimation strategy that requires a time window that scales as ${\mathcal O}(\log{1/\epsilon})$.
This corresponds to an exponential improvement in the circuit depth dependence on the accuracy.

The key observation for the design of low-depth filter functions is as follows. If a filter function satisfies the following properties, then it can isolate the minimum eigenvalue from the others well and the corresponding convolution can be evaluated easily:
\begin{enumerate}
    \item The filter function $f(x)$ has an exponentially-decaying tail, i.e., $|f(x)|=\exp{-\Omega(|x|)}$ for sufficiently large $x$. {This ensures that $(f*p)(x) \approx p_0 f(x-E_0)$ for $x$ near $E_0$ (i.e. the interference from the excited states is negligible) and hence we can easily infer $E_0$ from the shape of $f*p$ in this region}.
    \item The filter function's Fourier transform $\hat{f}(t)$ also has an exponentially-decaying tail, i.e., $|\hat{f}(t)|=\exp{-\Omega(|t|)}$ for sufficiently large $t$. This allows $f$ to be well-approximated by a band-limited function, which means that the maximal evolution time in the Hadamard tests will be small. 
\end{enumerate}

{Based on this observation, a natural choice is the Gaussian filter, defined as:
\begin{align}
    f_\sigma(x)=\frac{1}{\sqrt{2\pi}\sigma} e^{-\frac{x^2}{2\sigma^2}},
\end{align}
where $\sigma>0$ is a parameter to be chosen later.
Note that its Fourier transform is another Gaussian kernel (up to some scaling factor):
\begin{align*}
    \hat{f}_\sigma(t)=e^{-\frac{1}{2}(\sigma \pi t)^2}.
\end{align*}
Thus, most of its mass is concentrated within $|t|=\mathcal{O}(\sigma^{-1})$. More importantly, by convolving $f_\sigma$ with the spectral measure $p$, we get:
\begin{align*}
    (f_\sigma \ast p)(x) = \sum_{j=0}^{N-1} p_j f_\sigma (x-E_j),
\end{align*}
which is a mixture of Gaussians. Since the Gaussian filter has an exponentially-decaying tail, if we zoom-in to a neighborhood of $E_0$, the convolution value is dominated by the first Gaussian kernel $p_0 f_\sigma(x-E_0)$. Therefore, the first significant \emph{peak} of $f_\sigma \ast p$ will be close to $E_0$ and GSEE is then reduced to a \emph{peak finding} problem. Our approach to this problem is to first obtain some $\tilde{E}_0\in [E_0-{\cal O}(\sigma), E_0+{\cal O}(\sigma)]$ by using a previous algorithm \cite{lin2022heisenberg}. Then we partition the interval $[\tilde{E}_0-{\cal O}(\sigma), \tilde{E}_0+{\cal O}(\sigma)]$ into a  ${\cal O}(\epsilon)$-width grid, and estimate the convolution $f_\sigma\ast p$ at each grid point. Finally, we output the position of the grid point with maximum convolution value. The complexity of this algorithm depends on $\sigma$, the width of the Gaussian filter, since we can only truncate its spectrum to $[-T, T]$ for $T=\tilde{\Theta}(1/\sigma)$ in order to evaluate the convolution with enough precision. For sufficiently small $\epsilon>0$, one can take $\sigma = {\cal O}(\Delta / \mathrm{polylog}(\Delta\epsilon^{-1}\eta^{-1}))$
such that the algorithm outputs an estimate of $E_0$ within $\epsilon$-additive error. It implies that the maximal Hamiltonian evolution time of our algorithm is $\tilde{\cal O}(1/\Delta)$. }

{In Appendix \ref{sec:general_eval_conv}, we develop a method for evaluating the convolution of the spectral measure and any filter with bounded band-limit. We prove that if $f_T$ is a function such that $\hat{f}_T$ has support in $[-T, T]$, then one can use the measurement outcomes of $\tilde{\cal O}(\epsilon_1^{-2}\|\hat{f}_T\|_1^2)$ Hadamard tests to estimate $(f_T \ast p)(x)$ within accuracy $\epsilon_1>0$ (for any given $x$), where  $\|\hat{f}_{T}\|_{1}$ is the $L^1$ norm of $\hat{f}_{T}$. }

{We now apply this result to bound the total evolution time of our algorithm with the Gaussian filter. As we will see, its performance will be sub-optimal, so we will only analyze a more refined version with better performance in detail. For the truncated Gaussian filter 
\begin{align*}
    f_{\sigma, T}(x):=\int_{-T}^T \hat{f}_\sigma(t) e^{2\pi \i x t}{\rm d}t,
\end{align*}
we have $\|\hat{f}_{\sigma,T}\|_1={\cal O}(\sigma^{-1})$, and we need to evaluate the convolution within accuracy $\epsilon_1=\tilde{\cal O}(\eta \epsilon^2 \sigma^{-3})$ \footnote{{To see this, note that the Taylor expansion for the Gaussian density around $0$ up to second order yields
\begin{align*}
 f_\sigma(\epsilon)\simeq \frac{1}{\sqrt{2\pi}\sigma}-\frac{\epsilon^2}{2\sqrt{2\pi}\sigma^3}.   
\end{align*}
This implies that the difference between the values of $(f_\sigma*p)(x)\approx p_0 f_\sigma(x-E_0)$ for $x=E_0$ and $x=E_0+\epsilon$ can be as small as $\mathcal{O}(\eta\epsilon^2\sigma^{-3})$.}}. Thus, by our choice of $\sigma=\tilde{\cal O}(\Delta)$, the sample complexity is  $\tilde{\cal O}(\eta^{-2}\epsilon^{-4} \sigma^6\cdot \sigma^{-2})=\tilde{\cal O}(\eta^{-2}\epsilon^{-4} \Delta^{4})$. Hence, 
the total evolution time is ${\cal T}_{\sf tot}\leq \tilde{\cal O}(\eta^{-2}\epsilon^{-4} \Delta^{4} \cdot T)=\tilde{\cal O}(\eta^{-2}\epsilon^{-4} \Delta^{3})$ as $T=\tilde{\cal O}(1/\sigma)=\tilde{\cal O}(1/\Delta)$.}

{The bottleneck of our total evolution time is the \emph{normalized} convolution evaluation accuracy ${\epsilon_1}/{\|\hat{f}_{T}\|_1}$ which scales as ${\cal O}(\eta \epsilon^2 \sigma^{-2})$ for the Gaussian filter $f_{\sigma}$. To improve this factor, we switch to the Gaussian derivative filter $g_\sigma$ which is defined as follows:
\begin{align*}
    g_\sigma (x):=-\frac{1}{\sqrt{2\pi} \sigma^3 } x e^{-\frac{x^2}{2\sigma^2}}.
    \label{eq:gauss_deriv}    
\end{align*}
Since the Gaussian derivative filter has an exponentially-decaying tail, $(g_\sigma * p)(x)$ resembles $p_0 g_\sigma(x-E_0)$ in a neighborhood of $E_0$. In particular, the unique zero point of $g_\sigma * p$ in this region is close to $E_0$.}

{The Gaussian derivative filter allows for a more favorable normalized convolution evaluation accuracy. On the one hand, the difference between the values of $|(g_\sigma\ast p)(x)|$ for $x$ that is $\epsilon/2$-close to $E_0$ and for $x$ that is $\epsilon$-far from $E_0$ is $\Omega(\eta \epsilon \sigma^{-3})$ (see Lemma \ref{lem:convolution_separation}). So it suffices to pick $\epsilon_1={\cal O}(\eta \epsilon \sigma^{-3})$. On the other hand, it is easy to show that $\|\hat{g}_{\sigma, T}\|_1 = \Theta(\sigma^{-2})$. This implies that the required normalized convolution evaluation accuracy for $g_\sigma$ is ${\epsilon_1}/{\|\hat{g}_{\sigma,T}\|_1}={\cal O}(\eta \epsilon \sigma^{-1})$. Moreover, our GSEE and convolution evaluation approaches are general so that they can be easily adapted to the Gaussian derivative filter function with almost the same parameters (i.e., $\sigma=\tilde{\cal O}(\Delta)$ and $T=\tilde{\cal O}(1/\sigma)$). Therefore, using $g_\sigma$ in our algorithm, the maximal evolution time remains to be ${\cal T}_{\sf max}=\tilde{\cal O}(\Delta^{-1})$ and the total evolution time is reduced to ${\cal T}_{\sf tot}=\tilde{\cal O}(\eta^{-2} \epsilon^{-2}\Delta)$. }

\paragraph{Ground state energy estimation algorithm}
Using the Gaussian derivative filter, we describe the algorithm for ground state energy estimation that proves Theorem \ref{thm:gsee_main_informal}.
This is the first GSEE algorithm that uses $\tilde{\cal O}(\Delta^{-1})$-depth quantum circuits to achieve accuracy $\epsilon$.
A detailed presentation of the algorithm can be found in Algorithm \ref{alg:gsee}. {This algorithm makes use of the data structure 
for convolution evaluation in Algorithm \ref{alg:eval_conv}.}

\begin{algorithm}[ht!]
\caption{Convolution evaluation data structure.}
\label{alg:eval_conv}
\begin{algorithmic}[1]
\State {\bf data structure} \textsc{FilterSampler}
\State \hspace{4mm} \textsc{Init}($f_T$)\Comment{Initialize for the filter $f_{T}$}
\State \hspace{4mm} \textsc{Sample}()\Comment{Sample $\xi\in \R$ with probability $\propto |\hat{f}_{T}(\xi)|$}
\State \hspace{4mm} \textsc{Norm}()\Comment{Return $\|\hat{f}_{T}\|_1$}
\State {\bf end data structure}
\State
\State {\bf data structure} \textsc{ConvEval}
\State {\bf members}
    \State \hspace{4mm} ${\cal C}({H, \rho}, t, W)$ \Comment{Run the circuit in Figure \ref{fig:hadamard_test} with $\tau=2\pi t$ and $W=I$ or $S^\dagger$}
    \State \hspace{4mm} $\{(t^{(i)}, z^{(i)})\}_{i\in [S]}\subset \R\times\C$\Comment{Fourier samples}
    \State \hspace{4mm} \textsc{FilterSampler} FS\Comment{Filter function's sampler}
\State {\bf end members}\\
\Procedure{Init}{{$H$, $\rho$}, $f_T$, $\epsilon$, $\delta$, $M$}~~~\Comment{$\epsilon$ is the target  accuracy, $\delta$ is the tolerable failure probability, $M$ is the maximal number of points at which the convolution is evaluated}
    \State FS.\textsc{Init}($f_T$)
    \State $L\gets \mathrm{FS}.\textsc{Norm}()$
    \State $S\gets \left\lceil \frac{L^2 \ln{4M/\delta}}{\epsilon^2} \right\rceil$
    \Comment{Lemma~\ref{lem:sampling_overhead}}
    \For{$i\gets 1,2,\ldots,S$}
        \State $t^{(i)}\gets \mathrm{FS}.\textsc{Sample}()$
        \State $x^{(i)}\gets {\cal C}({H, \rho}, t^{(i)}, I) $\Comment{Hadamard test}
        \State $y^{(i)}\gets {\cal C}({H, \rho}, t^{(i)}, S^\dagger) $\Comment{Hadamard test}        
        \State $z^{(i)}\gets L\cdot e^{2\pi\i \phi(t^{(i)})}(x^{(i)}+\i y^{(i)})$
    \EndFor
\EndProcedure\\
\Procedure{Eval}{$x$}\Comment{Approximate $(f_T\ast p)(x)$ within accuracy $\epsilon$}
    \State $\overline{Z}\gets \frac{1}{S}\sum_{i\in [S]}e^{2\pi\i t^{(i)} x}\cdot z^{(i)}$
    \State \Return $\overline{Z}$
\EndProcedure
\State {\bf end data structure}
\end{algorithmic}
\end{algorithm}

\begin{algorithm}[ht!]
\caption{Low-depth ground state energy estimation algorithm.}\label{alg:gsee}
\begin{algorithmic}[1]
\Procedure{{GSEE}}{{$H, \rho, \epsilon, \delta, \Delta, \eta$}}
    \State  $\sigma \gets \min\lrb{\frac{0.9 \Delta}{\sqrt{2 \ln{9 \Delta \epsilon^{-1} \eta^{-1}}}}, 0.2 \Delta}$ 
    \State {Run the algorithm in \cite{lin2022heisenberg} on $H$ and $\rho$ to obtain an estimate $\tilde{E}_0$ of $E_0$ such that $\tilde{E}_0$ is $\sigma/4$-close to $E_0$ with probability at least $1-\delta/2$} \label{ln:get_coarse_estimate}    
    \State $M\gets \lceil \sigma / \epsilon\rceil + 1$
    \State $\tilde{\epsilon} \gets \frac{0.1 \epsilon  \eta}{\sqrt{2\pi}\sigma^3}$
    \State $T\gets \pi^{-1} \sigma^{-1} \sqrt{2 \ln{8 \pi^{-1} \tilde{\epsilon}^{-1} \sigma^{-2}  } }$\Comment{Filter band-limit (Lemma~\ref{lem:gaussian_derivative_bandlimit})}
    \State \textsc{ConvEval.Init}({$H$, $\rho$}, $g_{\sigma, T}$, $\tilde{\epsilon}$/2, $\delta/2$, $M$)\Comment{Algorithm~\ref{alg:eval_conv}}\label{ln:call_conv_init}
    \For{$j=1,2,\dots,M$}\label{ln:forloop_gsee}
        \State $x_j\gets \wt{E}_0 - 0.25 \sigma +(0.5\sigma/M)\cdot (j-1)$\label{ln:def_x_j}
        \State $h_j\gets \textsc{ConvEval.Eval}(x_j)$\Comment{Algorithm~\ref{alg:eval_conv}}
    \EndFor
    \State $j^*\gets \arg\min_{1\le j\le M} |h_j|$. 
    \State \Return $x_{j^*}$
\EndProcedure
\end{algorithmic}
\label{alg:low_depth_gsee}
\end{algorithm}

{The inputs to Algorithm \ref{alg:low_depth_gsee} include the Hamiltonian $H$ and initial state $\rho$, a lower bound $\Delta$ on the spectral gap $E_1-E_0$ of $H$, a lower bound $\eta$ on the overlap between $\rho$ and the ground state $\ket{E_0}$ of $H$ (i.e. $\eta \le \bra{E_0} \rho \ket{E_0}$),
the target accuracy $\epsilon$ and confidence $1-\delta$. The output of the algorithm is an estimate of the ground state energy $E_0$ of $H$.}
The steps of the algorithm are as follows:
\vspace{0.25cm}

\begin{enumerate}
    \item {\textbf{Roughly locate $E_0$:} Run the algorithm of \cite{lin2022heisenberg} to obtain an estimate $\tilde{E}_0$ of $E_0$ such that $|\tilde{E}_0-E_0|=\wt{O}(\Delta)$ with high probability.}
    \item \textbf{Configure filter:} Set the Gaussian derivative filter function to have width $\sigma \in \wt{O}(\Delta)$ and choose the truncation of the filter to be $T\in \wt{O}(1/\sigma)=\wt{O}(1/\Delta)$ (this limits the circuit depth).
    \item \textbf{Estimate convolution:} 
    For a grid of $M \in \wt{O}(\Delta/\epsilon)$ evenly-spaced energies centered at $\wt{E}_0$ with width $\wt{O}(\Delta)$, estimate the Gaussian derivative convolution at each grid point.
    \begin{enumerate}
    \item To estimate the Gaussian derivative convolution at each point, for $\wt{O}(\eta^{-2}\epsilon^{-2}\Delta^2)$ rounds, draw a time $t\in[-T, T]$ with probability proportional to $|\hat{f}_T(t)|$ and run depth $t$ real and imaginary Hadamard tests to generate binary samples $X$ and $Y$. 
    \item Compute $Z(x;t,X,Y)$ (see Eq. \ref{eq:def_z}) for each sample and average to output the convolution estimate at $x$.
    \label{step:conv_est}
    \end{enumerate}
    \item \textbf{Estimate zero-crossing:} Among the $M$ convolution estimates, find the estimate closest to zero and report the corresponding energy as the ground state energy estimate.
\end{enumerate}
To realize the results in Corollary \ref{cor:interpolation_informal}, we choose $\Delta$ between $\Delta_{\textup{true}}$ and $1/\epsilon$.

{Furthermore, we can easily parallelize our algorithm to reduce the runtime in the regime where $\Delta\gg \epsilon$. Indeed, as we do a Monte Carlo evaluation of the convolution, we can use several quantum computers in parallel to generate samples and reduce the runtime of the computation. In contrast, in the regime $\Delta\simeq \epsilon$ it is unclear how to speed up the computation by resorting to parallel quantum computers, as we essentially sample from the distribution a constant number of times. Thus, as the total number of gates required to implement the algorithm in the $\Delta\gg \epsilon$ regime is smaller than the usual QPE, we envision that it will be possible to run the algorithm reliably in error-corrected devices with a smaller code distance, which translates to a smaller number of physical qubits. And by running the algorithm on various smaller quantum processors in parallel, it will be possible to offset some of the additional cost incurred by the quadratically worse dependency on the precision. By combining these two observations, we believe that it will be possible to obtain quantum advantage earlier with the approach advocated by this paper when compared to traditional QPE.}

 {One may wonder if Algorithm \ref{alg:low_depth_gsee} still works if $\Delta$ is larger than the spectral gap $E_1-E_0$. Unfortunately, this is not always the case. To see this, recall that we infer $E_0$ from the shape of the convolution of the spectral measure $p$ and a Gaussian derivative filter $g_\sigma$. We need the filter $g_\sigma$ to be narrow enough so that $(p*g_\sigma)(x) \approx p_0 g_\sigma(x-E_0)$ in a neighborhood of $E_0$. Then the unique zero point of $p*g_\sigma$ in this region is close to $E_0$. If $\Delta > E_1-E_0$, then since $\sigma=\tilde{\cal O}(\Delta)$, the filter $g_\sigma$ could be too wide for our purpose, and the functions $p_0 g_\sigma(x-E_0)$ and $p_1 g_\sigma(x-E_1)$ might have significant overlap. Consequently, we can no longer guarantee that $(p*g_\sigma)(x) \approx p_0 g_\sigma(x-E_0)$ for $x$ near $E_0$. In other words, $p*g_\sigma$ might have a different shape from $g_\sigma$ in the neighborhood of $E_0$. Thus, the current strategy  will not work in general. In practice, it may be difficult to decide whether a given $\Delta$ is smaller or larger than $E_1-E_0$, and this could pose a problem for the application of Algorithm \ref{alg:low_depth_gsee}. In a recent work \cite{wang2023faster}, we proposed a method to certify the correctness of the outputs of GSEE algorithms (without assuming $\Delta\le E_1-E_0$) to mitigate this problem. We leave it as future work to fully resolve this issue.}

\section{Future directions}

In this work, we have introduced a framework for estimating ground state energies using tunable-depth quantum circuits.
As shown in Figure \ref{fig:intro}, the algorithms developed in this work are applicable to maximum circuit depths ranging from $\wt{\cal O}(1/\Delta)$ to $\wt{\cal O}(1/\epsilon)$.
The development of algorithms beyond this range is left for future exploration.
While we have made progress in establishing upper bounds over a range of circuit depths, an important future direction is to establish lower bounds on depth-limited ground state energy estimation.
This would deepen our understanding of the capabilities and constraints of using depth-limited quantum computers for simulating quantum systems. {To this end, one might draw inspiration from Section 2.2 of \cite{dutkiewicz2022heisenberg} which derives Cramer-Rao lower bounds on the scaling of the error versus the total quantum cost for estimating a single eigenvalue phase given an eigenstate input. It might be possible to generalize the result to the case where the input state is just close to an eigenstate and a spectral gap of the Hamiltonian is promised.}

Our work helps to establish the paradigm of developing quantum algorithms using the tools of classical signal processing \cite{lin2022heisenberg, zhang2021computing, wang_boosters}. 
In this approach, the quantum computer produces stochastic signals that encode characteristics of a matrix.
This stochastic signal can be processed to learn the matrix properties of interest.
The signal processing paradigm is well-suited to
developing algorithms for early fault-tolerant quantum computers.
Quantum computations with such architectures will be error prone, generating noisy signals.
The tools of classical signal processing have been designed to handle such noisy signals and can aid in the design and analysis of robust quantum algorithms \cite{wang2021minimizing, katabarwa2021reducing, kshirsagar2022proving}.

One requirement of the algorithm is that a lower bound on the energy gap must be specified.
There exist quantum chemistry methods for estimating the gap
(e.g. using the ORCA software \cite{neese2012orca,neese2018software} as we did for our numerical comparisons).
However, such estimates can become inaccurate for large systems.
It may be helpful to incorporate a step into the quantum algorithm that estimates this gap.
Although this estimation is computationally hard in general \cite{ambainis2014physical}, many physical systems of interest have structure that make the estimation feasible.

In this work, we did not consider the impact of implementation error on the performance of the algorithm.
We expect that our algorithm is able to tolerate some degree of variation between the ideal Hadamard tests and the implemented Hadamard tests.
Building off of recent work \cite{kshirsagar2022proving}, we believe the algorithm can be operated so as to accommodate such deviations.
We leave for future work the investigation of robust quantum algorithms for ground state energy estimation.

The methods introduced here may help to bring the target of useful quantum computing closer to the present.
Yet, there is still much work needed to carry out detailed resource estimations that predict the onset of quantum advantage using methods such as those we have introduced.
More broadly, our hope is that this work contributes to the general understanding of how to use quantum computers given practical constraints on their capabilities and might inspire the development of quantum algorithms in other application domains.

{One may have noticed that GSEE can be also solved by preparing a high-fidelity approximation of the ground state and estimating its energy with respect to the Hamiltonian somehow. However, it is unclear whether such methods can have the same circuit depth and quantum runtime as ours. Specifically, our method requires only $\tilde{\mathcal O}(\Delta^2/\epsilon^2)$ quantum circuits where each circuit evolves $H$ for $\tilde{\mathcal O}(1/\Delta)$ time (ignoring $\eta$-dependence). Note that $\Delta \le E_1-E_0$ is always no larger than $\norm{H}$. We will consider two state-preparation-based strategies below, and both of them require $\Omega(\norm{H}^2/\epsilon^2)$ copies of the state to reach accuracy $\epsilon$, which is more costly than our method, especially when $\norm{H} \gg \Delta$. In both strategies, let $\ket{\psi}$ be the output of a ground state preparation algorithm.
\begin{enumerate}
    \item In the first strategy, suppose $H$ has the Pauli decomposition $H=\sum_i a_i P_i$. Then with probability $|a_i|/\norm{\vec{a}}_1$, where $\norm{\vec{a}}_1=\sum_i |a_i| \ge \norm{H}$, we perform the projective measurement $\{(I+P_i)/2, (I-P_i)/2\}$ on $\ket{\psi}$, and multiply the outcome $\pm 1$ by $a_i/|a_i|$. Then this random variable is an unbiased estimator of $\bra{\psi}H\ket{\psi}/\norm{\vec{a}}_1$. Thus we need to collect ${\mathcal O}(\norm{\vec{a}}_1^2 /\epsilon^2)$ samples to estimate $\bra{\psi}H\ket{\psi}$ within accuracy $\epsilon$.
    \item In the second strategy, let $U$ be a block-encoding of $H$ with normalization factor $\alpha$, i.e. $\bra{0}U\ket{0} = H/\alpha$. Note that $\alpha \ge \norm{H}$. Then we run the Hadamard test with unitary operation $U$ and initial state $\ket{0}\ket{\psi}$. Then the outcome of this circuit is an unbiased estimator of $\bra{0}\bra{\psi}U\ket{0}\ket{\psi}=\bra{\psi}H\ket{\psi}/\alpha$. Hence we need to collect ${\mathcal O}(\alpha^2/\epsilon^2)$ samples to estimate $\bra{\psi}H\ket{\psi}$ within accuracy $\epsilon$.
\end{enumerate}
We leave it as an open question to find a state-preparation-based strategy with $\tilde{\mathcal O}(\Delta^2 /\epsilon^2)$ sample complexity.}

{In this work, we have focused on estimating a single eigenvalue of the Hamiltonian $H$. But it is likely that our method can be extended to estimate multiple eigenvalues of $H$ simultaneously, under appropriate assumptions about the gaps among the eigenvalues and the overlaps between an input state and the eigenstates. The reason is as follows. One can use the data from the Hadamard tests to construct the convolution of the spectral measure and a Gaussian filter, and under proper conditions, the peaks of this convolution will be close to the eigenvalues of $H$. Then our problem is reduced to finding the peaks of this function. Furthermore, perhaps we can replace the Gaussian filter by a Gaussian derivative filter to gain better efficiency, as in GSEE. In that case, one would need to search for the zero points of the convolution instead of its peaks. We leave it as future work to fully develop this algorithm for simultaneous estimation of multiple eigenvalues of a given Hamiltonian. }

\vspace{0.5cm}

\noindent\textbf{Acknowledgements} This work was done while R.Z. and S.Z. were research interns at Zapata Computing Inc. We especially thank Alex Kunitsa for carrying out the gap estimations for the EC and PF$_6^-$ molecules and providing quantum chemistry expertise on the topic and we thank Mario Berta for the suggestion to include such estimations in the manuscript. We thank Chong Sun for developing the software used to generate Figure \ref{fig:GSEE_circuits}. We also thank Yudong Cao, Phillip Jensen, Artur Izmaylov, Yu Tong, and Katerina Gratsea for feedback on the manuscript. Finally, we thank Aram Harrow and Andrew Childs for discussions and suggestions on future directions for this line of research.

\appendix

\section{Estimating ground state energy via Gaussian derivative filtering}

In this appendix, we propose a strategy for GSEE based on Gaussian derivative filtering. In Appendix~\ref{sec:filter_prop}, we define the Gaussian derivative function and prove a nice property of the convolution between this filter and the spectral measure $p$. In Appendix~\ref{sec:gsee_algorithm}, we show how this property leads to a strategy for GSEE. In Appendix~\ref{sec:filter_truncation}, we prove that the Gaussian derivative function can be approximated by a band-limited function, which is crucial for efficient evaluation of the convolution. 

\subsection{Convolving the spectral measure with a Gaussian derivative filter}
\label{sec:filter_prop}
Let us start by defining the Gaussian derivative function and demonstrating its properties. Specifically, let $\sigma>0$ be arbitrary, and let $f_{\sigma}(x)=\frac{1}{\sqrt{2\pi} \sigma} e^{-\frac{x^2}{2\sigma^2}}$ be a Gaussian function. The Fourier transform of $f_\sigma$ is
\begin{align}
    \hat{f}_\sigma(\xi) = e^{-\frac{1}{2} (\sigma \pi \xi)^2}.
\end{align}
Now consider the derivative of $f_\sigma$, i.e.,
 \begin{align}\label{eqn:g_sigma}
     g_\sigma(x) \defeq f'_\sigma(x) = -\frac{1}{\sqrt{2\pi} \sigma^3 } x e^{-\frac{x^2}{2\sigma^2}}.
 \end{align}
Then the Fourier transform of $g_\sigma$ is
\begin{align}
    \hat{g}_\sigma(\xi) = 2\pi \i \xi \hat{f}_\sigma(\xi)
    = 2 \pi \i \xi e^{-\frac{1}{2} (\sigma \pi \xi )^2}.
\end{align}

The following properties of $g_\sigma$ and $\hat{g}_\sigma$ will be useful:
\begin{fact}[Properties of the Gaussian derivative function]\label{fact:guassian_derivative}
~
\begin{enumerate}
    \item $g_\sigma(0)=0$.
    \item $|g_\sigma(x)|$ is even, increases monotonically in $(-\infty, -\sigma] \cup [0, \sigma]$, and decreases monotonically in $[-\sigma, 0] \cup [\sigma, \infty)$.       
    \label{monotonicity_of_filter}
    \item $g_\sigma(x)$ decays exponentially to $0$ as $x \to \pm\infty$.
    \item $\hat{g}_\sigma(\xi)$ decays exponentially to $0$ as $\xi \to \pm\infty$.
\end{enumerate}

\end{fact}

Now let us consider the convolution between the filter $g_\sigma$ and the spectral measure $p$: 
\begin{align}
    (g_\sigma * p) (x) = \sum_{j=0}^{N-1} p_j g_\sigma(x-E_j)
    = -\frac{1}{\sqrt{2\pi} \sigma^3 } \sum_{j=0}^{N-1} p_j 
     (x-E_j) e^{-\frac{(x-E_j)^2}{2\sigma^2}}. 
     \label{eqn:convolution}
\end{align}
It turns out that if $\sigma$ is appropriately chosen, then $|(g_\sigma * p)(x)|$ is small only if $x$ is close to $E_0$, assuming $x$ is at most $O(\sigma)$-away from $E_0$:

\begin{lemma}\label{lem:conditions_gaussian_filter}
Let $c=\sqrt{2\ln{10/9}}\approx 0.45904$, and let $\Delta$ and $\eta$ be as in the problem formulation in the main text. Suppose $\epsilon>0$ is small enough such that $\epsilon \le c \cdot \min \left(\frac{0.9  \Delta}{\sqrt{2\ln{9 \Delta \epsilon^{-1} \eta^{-1}}}}, 0.2  \Delta \right)$. Then for
\begin{align}
    \sigma \defeq \min\lrb{\frac{0.9 \Delta}{\sqrt{2 \ln{9 \Delta \epsilon^{-1} \eta^{-1}}}}, 0.2 \Delta}, 
    \label{eq:def_sigma}
\end{align} 
we have 
\begin{itemize}
    \item $|(g_\sigma * p)(x)| < \dfrac{0.6 \epsilon p_0}{\sqrt{2\pi} \sigma^3} $, $\forall x \in [E_0-0.5\epsilon, E_0+0.5\epsilon]$.
    \item $|(g_\sigma * p)(x)| > \dfrac{0.8 \epsilon p_0}{\sqrt{2\pi} \sigma^3} $, $\forall x \in [E_0-0.5\sigma, E_0-\epsilon) \cup (E_0+\epsilon, E_0+0.5\sigma]$.
\end{itemize}
\label{lem:convolution_separation}
\end{lemma}

\begin{proof}
Note that our choice of $\sigma$ and the condition on $\epsilon$ imply that $\epsilon \le c \sigma < 0.5\sigma$. As a consequence, we do have $E_0-0.5\sigma < E_0-\epsilon$ and $E_0+\epsilon < E_0+0.5\sigma$. Thus, the interval in the second bullet is well-defined.  Moreover, we have
\begin{align}
\abs{g_\sigma(0.9 \Delta)}
    &= \frac{1}{\sqrt{2\pi} \sigma^3} 0.9 \Delta e^{-\frac{0.81\Delta^2}{2 \sigma^2}} \nonumber \\
&\le \frac{1}{\sqrt{2\pi} \sigma^3} 0.1 \epsilon \eta \tag{by the property $\sigma\leq \frac{0.9 \Delta}{\sqrt{2 \ln{9 \Delta \epsilon^{-1} \eta^{-1}}}}$ in Eq.~\eqref{eq:def_sigma}} \\ 
& \le \frac{1}{\sqrt{2\pi} \sigma^3} 0.1 \epsilon p_0,
\label{eq:bound_tail_impact}
\end{align}
where the last step follows from $p_0\geq \eta$.

We prove the first and the second parts of the lemma below.

\paragraph{Part I.} For any $x \in [E_0-0.5 \epsilon, E_0+0.5 \epsilon]$, we have
 \begin{align}
    \abs{(g_\sigma*p)(x)} &= \abs{p_0 g_\sigma(x-E_0) + \sum_{j=1}^{N-1} p_j g_\sigma(x-E_j)} \tag{by Eq.~\eqref{eqn:convolution}} \nonumber \\
    & \le p_0 \abs{g_\sigma(x-E_0)} + \sum_{j=1}^{N-1} p_j \abs{g_\sigma(x-E_j)} \nonumber \\ 
& \le p_0 \abs{g_\sigma(x-E_0)} + \max_{1 \le j \le N-1}\abs{g_\sigma(x-E_j)}.\label{eqn:convolution_upperbound}
\end{align}

The first term in Eq.~\eqref{eqn:convolution_upperbound} can be bounded as follows: 
\begin{align}
\abs{g_\sigma(x-E_0)} &\le \abs{g_\sigma(0.5 \epsilon)} \tag{by $|x-E_0|\leq 0.5\epsilon<\sigma$ and Property \ref{monotonicity_of_filter} in Fact~\ref{fact:guassian_derivative}} \\
&= \frac{1}{\sqrt{2\pi} \sigma^3} 0.5 \epsilon e^{-\frac{0.25\epsilon^2}{2 \sigma^2}} \tag{by Eq.~\eqref{eqn:g_sigma}} \\
&\le \frac{1}{\sqrt{2\pi} \sigma^3} 0.5 \epsilon.\label{eq:part1_1}
\end{align} 

To upper bound the second term in Eq.~\eqref{eqn:convolution_upperbound}, first note that for each $j \ge 1$, 
\begin{align}
|x-E_j| &\ge E_j-E_0-0.5\epsilon \tag{since $x\in[E_0-0.5\epsilon, E_0+0.5\epsilon]$} \\
&\ge \Delta - 0.5\epsilon \tag{since $E_j-E_0 \ge E_1-E_0 \ge \Delta$}\\
&> 0.9 \Delta  \tag{by the assumption $\epsilon \le 0.2 c \Delta < 0.1 \Delta$}\\
&> \sigma,\label{eqn:x-E_i}
\end{align}
where the last step follows from the property $\sigma \le 0.2 \Delta$ in Eq.~\eqref{eq:def_sigma}. Then we obtain 
\begin{align}
    \abs{g_\sigma(x-E_j)} &< \abs{g_\sigma(0.9 \Delta)}\tag{by Eq.~\eqref{eqn:x-E_i} and Property \ref{monotonicity_of_filter} in Fact~\ref{fact:guassian_derivative}}\\
&\le \frac{1}{\sqrt{2\pi} \sigma^3} 0.1 \epsilon p_0, \label{eq:part1_2}
\end{align}
where the second step follows from Eq.~\eqref{eq:bound_tail_impact}. 

Combining Eqs.~\eqref{eqn:convolution_upperbound}, \eqref{eq:part1_1}, and \eqref{eq:part1_2}, we get that for $x \in [E_0-0.5 \epsilon, E_0+0.5 \epsilon]$,
\begin{align}
    \abs{(g_\sigma*p)(x)} &< p_0 \cdot \frac{1}{\sqrt{2\pi} \sigma^3} 0.5 \epsilon
    + \frac{1}{\sqrt{2\pi} \sigma^3} 0.1 \epsilon p_0
    = \frac{0.6 \epsilon p_0}{\sqrt{2\pi} \sigma^3}.
\end{align}

\paragraph{Part II.} For any $x \in [E_0-0.5\sigma, E_0-\epsilon) \cup (E_0+\epsilon, E_0+0.5\sigma]$, we have 
\begin{align}
    \abs{(g_\sigma*p)(x)} &= \abs{p_0 g_\sigma(x-E_0) + \sum_{j=1}^{N-1} p_j g_\sigma(x-E_j)} \tag{by Eq.~\eqref{eqn:convolution}} \\
    & \ge p_0 \abs{g_\sigma(x-E_0)} - \sum_{j=1}^{N-1} p_j \abs{g_\sigma(x-E_j)} \nonumber \\ 
    & \ge p_0 \abs{g_\sigma(x-E_0)} - \max_{1 \le j \le N-1} \abs{g_\sigma(x-E_j)}.\label{eqn:convolution_lower_bound}
\end{align}

The first term in Eq.~\eqref{eqn:convolution_lower_bound} can be lower bounded as follows:
\begin{align}
   \abs{g_\sigma(x-E_0)} &> \abs{g_\sigma(\epsilon)} \tag{by $\epsilon < |x-E_0| \le 0.5 \sigma$ and Property \ref{monotonicity_of_filter} in Fact~\ref{fact:guassian_derivative}}\\
&= \frac{1}{\sqrt{2\pi} \sigma^3} \epsilon e^{-\frac{\epsilon^2}{2 \sigma^2}} \tag{by Eq.~\eqref{eqn:g_sigma}}\\
&\geq \frac{1}{\sqrt{2\pi} \sigma^3} \epsilon e^{-\frac{c^2}{2}} \tag{by the assumption $\epsilon \le c\sigma$}\\
&\ge \frac{1}{\sqrt{2\pi} \sigma^3} 0.9 \epsilon,\label{eq:part2_1}
\end{align}
where the last step follows from $c=\sqrt{2\ln{10/9}}$. 

To upper bound the second term in Eq.~\eqref{eqn:convolution_lower_bound}, note that for each $j \ge 1$, 
\begin{align}
|x-E_j| &\ge E_j-E_0-0.5\sigma\tag{since $x \in [E_0-0.5\sigma, E_0-\epsilon) \cup (E_0+\epsilon, E_0+0.5\sigma]$}\\ 
&\ge \Delta - 0.5\sigma \tag{since $E_j-E_0 \ge E_1-E_0 \ge \Delta$}\\
&\ge 0.9 \Delta > \sigma, \label{eqn:x-E_j_lower_bound}
\end{align}
where the last two inequalities follow from the property $\sigma \le 0.2 \Delta$ in  Eq.~\eqref{eq:def_sigma}. Then we obtain
\begin{align}
\abs{g_\sigma(x-E_j)} &\le \abs{g_\sigma(0.9 \Delta)}\tag{by Eq.~\eqref{eqn:x-E_j_lower_bound} and Property \ref{monotonicity_of_filter} in Fact~\ref{fact:guassian_derivative}}\\
&\le \frac{1}{\sqrt{2\pi} \sigma^3} 0.1 \epsilon p_0,  \label{eq:part2_2}  
\end{align}
where the last step follows from Eq.~\eqref{eq:bound_tail_impact}. 

Combining Eqs.~\eqref{eqn:convolution_lower_bound}, \eqref{eq:part2_1}, and \eqref{eq:part2_2}, we get that for $x \in [E_0-0.5\sigma, E_0-\epsilon) \cup (E_0+\epsilon, E_0+0.5\sigma]$,
\begin{align}
    \abs{(g_\sigma*p)(x)} > p_0 \cdot \frac{1}{\sqrt{2\pi} \sigma^3} 0.9 \epsilon
    - \frac{1}{\sqrt{2\pi} \sigma^3} 0.1 \epsilon p_0 =
    \frac{0.8 \epsilon p_0}{\sqrt{2\pi} \sigma^3}.
\end{align}

The lemma is thus proved.
\end{proof}

\subsection{Basic strategy for ground state energy estimation}\label{sec:gsee_algorithm}
Lemma \ref{lem:convolution_separation} prompts us to develop the following strategy for estimating ground state energy. We first obtain an estimate $\tilde{E}_0$ of $E_0$ such that $\tilde{E}_0$ is $O(\sigma)$-close to $E_0$ with high probability. Then we find a point at which $|(g_\sigma * p)|$ has small value in a region $[\tilde{E}_0-O(\sigma), \tilde{E}_0+O(\sigma)]$. Using Lemma~\ref{lem:convolution_separation} we can prove that this point is $\epsilon$-close to $E_0$ with high probability. Formally, we have

\begin{lemma}
\label{lem:gsee_correctnesss}
Let $\Delta, \eta, \epsilon$ and $\delta$ be as in the problem formulation in the main text. Suppose $\epsilon$ satisfies the condition in Lemma \ref{lem:convolution_separation}. Let $\sigma$ be defined as Eq.~\eqref{eq:def_sigma}. Suppose $\tilde{E}_0$ is a random variable such that 
\begin{align}
\P{|\tilde{E}_0-E_0|>\frac{\sigma}{4}} < \frac{\delta}{2}.    
\label{eq:cond_wte0}
\end{align}
Let $M \defeq \lceil \sigma / \epsilon\rceil + 1$, and let $x_j \defeq \tilde{E}_0 - 0.25 \sigma +(0.5\sigma/M)\cdot (j-1)$ for $j \in [M]$. Suppose $h_1, h_2, \dots, h_M$ are random variables such that  
\begin{align}
\P{\forall j\in[M]:|(g_\sigma*p)(x_j) - h_j|\le \frac{0.1 \epsilon \eta}{\sqrt{2\pi}\sigma^3}} \ge 1-\frac{\delta}{2}.    
\label{eq:cond_hjs}
\end{align}
Let $j^*=\arg\min_{1\le j\le M} |h_j|$. Then we have
\begin{align}
    \P{|x_{j^*} - E_0| > \epsilon} < \delta.
\end{align}
\end{lemma}
\begin{proof}
By our assumptions about $\tilde{E_0}$ and $h_1, h_2, \dots,h_M$ and the union bound, we get that the following events happen simultaneously with probability at least $1-\delta$:
\begin{itemize}
    \item $|\tilde{E}_0-E_0| \le 0.25\sigma$.
    \item $|(g_\sigma*p)(x_j) - h_j|\le \frac{0.1 \epsilon \eta}{\sqrt{2\pi}\sigma^3}$, $\forall j \in [M]$.
\end{itemize}

In this case, we have $x_0, x_1, \dots, x_M \in [\tilde{E_0}-0.25\sigma, \tilde{E}_0+0.25\sigma] \subseteq [E_0 - 0.5 \sigma, E_0 + 0.5 \sigma]$. Then by Lemma \ref{lem:convolution_separation}, we have that
\begin{itemize}
    \item If $|x_j-E_0| \le 0.5\epsilon$, then
    \begin{align}
        |h_j| \le |(g_\sigma * p)(x_j)| + |(g_\sigma*p)(x_j) - h_j|
        < \frac{0.6 \epsilon p_0}{\sqrt{2\pi}\sigma^3} + \frac{0.1 \epsilon \eta}{\sqrt{2\pi}\sigma^3}
        \le \frac{0.7 \epsilon p_0}{\sqrt{2\pi}\sigma^3}.
    \end{align}
    \item If $|x_j-E_0| > \epsilon$, then
    \begin{align}
        |h_j| \ge |(g_\sigma * p)(x_j)| - |(g_\sigma*p)(x_j) - h_j|
        > \frac{0.8 \epsilon p_0}{\sqrt{2\pi}\sigma^3} - \frac{0.1 \epsilon \eta}{\sqrt{2\pi}\sigma^3}
        \ge \frac{0.7 \epsilon p_0}{\sqrt{2\pi}\sigma^3}.
    \end{align}    
\end{itemize}

Meanwhile, note that $x_1 \le E_0 \le x_M$, and $|x_{j+1}-x_{j}|\le 0.5\epsilon$, $\forall j \in [M-1]$. So there exists some $j_0 \in [M]$ such that $|x_{j_0}-E_0|\le 0.5\epsilon$. This implies that $|h_{j^*}| \le |h_{j_0}| < \frac{0.7 \epsilon p_0}{\sqrt{2\pi}\sigma^3}$, which in turn implies that $|x_{j^*}-E_0|\le \epsilon$. This lemma is thus proved.
\end{proof}

It remains to show how to generate the random variables $\tilde{E}_0$ and $h_1, h_2, \dots, h_M$ that satisfy the conditions Eqs.~\eqref{eq:cond_wte0} and \eqref{eq:cond_hjs} respectively. To obtain $\tilde{E}_0$, we use the GSEE algorithm in \cite{lin2022heisenberg} which takes $\tilde{O}(\epsilon^{-1})$ maximal Hamiltonian evolution time to achieve $\epsilon$-accuracy. Since $\tilde{E}_0$ only needs $\frac{\sigma}{4}$-accuracy, this step has $\tilde{O}(\sigma^{-1})$ maximal evolution time. To obtain $h_1, h_2, \dots, h_M$, we first introduce the band-limited version of $g_\sigma$, denoted as $g_{\sigma, T}$, in Appendix~\ref{sec:filter_truncation}, and prove that $(g_\sigma*p)(x) \approx (g_{\sigma,T}*p)(x)$ for a small $T$. Then we design a data structure \textsc{ConvEval} in Appendix~\ref{sec:complexity_analysis_eval} such that this data structure can evaluate $g_{\sigma,T}*p$ at the points $x_1,x_2,\dots,x_M$ with high accuracy and confidence after appropriate initialization.

\subsection{Gaussian derivative filters with bounded band-limits}
\label{sec:filter_truncation}
In order to efficiently evaluate $g_\sigma * p$ at any given point, we truncate the spectrum of $g_\sigma$ and construct a $T$-bandlimit version $g_{\sigma,T}$ such that
\begin{align}
    (g_\sigma * p)(x) \approx (g_{\sigma, T} * p)(x), ~~~\forall x \in \R.
\end{align}
Specifically, we define $g_{\sigma,T}$ by restricting $\hat{g}_\sigma$ to $[-T, T]$ and performing the inverse Fourier transform:
\begin{align}
    g_{\sigma, T}(x) \defeq  \int_{-T}^T \hat{g}_\sigma(\xi) e^{2\pi \i x \xi} d\xi.
\end{align}
Clearly, $g_{\sigma,T} \to g_{\sigma}$ as $T\to\infty$. The following lemma shows how to choose $T$ such that $g_{\sigma,T}$ can approximate $g_\sigma$ in $L_\infty$-norm:
\begin{lemma}\label{lem:gaussian_derivative_bandlimit}
Let $\epsilon_1>0$ be arbitrary. Then for 
\begin{align}
T \defeq \pi^{-1} \sigma^{-1} \sqrt{2 \ln{8 \pi^{-1} \epsilon_1^{-1} \sigma^{-2}  } },
\label{eq:def_T}
\end{align}
we have 
\begin{align}
|g_\sigma(x) - g_{\sigma, T}(x)| \le \frac{\epsilon_1}{2}, &~\forall x \in \R.    
\end{align}
\label{lem:upper_bound_on_T}
\end{lemma}
\begin{proof}
By the Fourier inversion theorem, we have
\begin{align}
\abs{g_\sigma(x) - g_{\sigma, T}(x)}
=& \abs{\int_{-\infty}^{-T} \hat{g}_{\sigma}(\xi) e^{2\pi \i \xi x} d\xi
+\int_{T}^{+\infty} \hat{g}_{\sigma}(\xi) e^{2\pi \i \xi x} d\xi} \nonumber\\
\le & \int_{-\infty}^{-T} \abs{\hat{g}_{\sigma}(\xi)} d\xi
+ \int_{T}^{+\infty} \abs{\hat{g}_{\sigma}(\xi)} d\xi \nonumber\\
=&2 \int_T^{+\infty} 2\pi \xi e^{-\frac{1}{2}(\sigma \pi \xi)^2} d\xi \nonumber\\
=& \frac{4}{\sigma^2 \pi} e^{-\frac{1}{2} \sigma^2 \pi^2 T^2}.
\end{align}
By solving the inequality
\begin{align}
    \frac{4}{\sigma^2 \pi} e^{-\frac{1}{2} \sigma^2 \pi^2 T^2} \le \frac{\epsilon_1}{2},
\end{align}
we get that it suffices to take
\begin{align}
T \ge  \pi^{-1} \sigma^{-1} \sqrt{2 \ln{8 \pi^{-1} \epsilon_1^{-1} \sigma^{-2}  } } .
\end{align}
\end{proof}

The following claim shows that the $L_\infty$-approximation for $g_\sigma$ implies the $L_\infty$-approximation for $g_\sigma\ast p$.
\begin{claim}\label{clm:l_infty_approx}
Let $T$ be defined as in Lemma \ref{lem:upper_bound_on_T}. Then we have 
\begin{align*}
    \abs{(g_\sigma * p)(x) - (g_{\sigma, T} * p)(x)} \leq \frac{\epsilon_1}{2}, ~&~\forall x\in \R.
\end{align*}
\end{claim}
\begin{proof}
For any $x\in \R$, we have
\begin{align}
    \abs{(g_\sigma * p)(x) - (g_{\sigma, T} * p)(x)}
    &=\abs{\int_{-\infty}^\infty  (g_\sigma(z) - g_{\sigma, T}(z)) p(x-z) dz} \nonumber\\
    & \le \int_{-\infty}^\infty  \abs{g_\sigma(z) - g_{\sigma, T}(z)} p(x-z) dz \nonumber\\
    & \le \frac{\epsilon_1}{2}     \int_{-\infty}^\infty  p(x-z) dz \nonumber\\
    & = \frac{\epsilon_1}{2},
\end{align}
where the first step follows from the definition of convolution, the second step follows from the triangle inequality, the third step follows from  Lemma \ref{lem:upper_bound_on_T}, and the last step follows from the property of Dirac delta function.
\end{proof}

Claim~\ref{clm:l_infty_approx} implies that in order to estimate $(g_\sigma \ast p)(x)$ within $ \epsilon_1$-accuracy, it suffices to evaluate $(g_{\sigma, T}\ast p)(x)$ within $0.5 \epsilon_1$-accuracy, which can be achieved by the method in Appendix \ref{sec:complexity_analysis_eval}.

\section{Complexity of evaluating the convolution}
\label{sec:complexity_analysis_eval}
In this appendix, we focus on evaluating the convolution between a filter function $f$ and the spectral measure $p$ to within $\epsilon$-additive error. In Appendix~\ref{sec:general_eval_conv}, we develop an evaluation method for general filter functions with bounded band-limits. Then in Appendix~\ref{subsec:cost_conv_eval_gdf}, we apply the method to the Gaussian derivative filter used in our GSEE algorithm.

\subsection{Evaluating the convolution via Hadamard tests}\label{sec:general_eval_conv}

For the sake of generality, we will not restrict to a specific filter $f$ but consider arbitrary filters with bounded band-limits. Specifically, for a parameter $T>0$, let $f_T$ be a function with band-limit $T$, i.e.,
\begin{align}
     f_{T}(x)=\int_{-T}^T\hat{f}_T(t)e^{2\pi \i t x}dt,    
\end{align}
where $\hat{f}_T$ is the Fourier transform of $f_T$ and satisfies $\hat{f}_T(t)=0$ for all $t \in (-\infty, -T) \cup (T, +\infty)$. Furthermore, we require that $\hat{f}_T$ is either continuous in $[-T, T]$ or a weighted sum of Dirac delta functions (i.e., $f_T$ has a discrete spectrum). Here we will state the results for the former case, and the reader can easily generalize them to the latter case.

Given such a function $f_T$, we can define a probability density $\nu$ in terms of its Fourier weights: 
\begin{align}\label{equ:definition_nu}
\nu(t)=\frac{|\hat{f}_{T}(t)|}{\|\hat{f}_{T}\|_{1}},~~~\forall t\in [-T,T].
\end{align}
Moreover, let $\phi(t)$ be the phase of $\hat{f}_T(t)$, i.e., $\hat{f}_T(t)=|\hat{f}_T(t)|e^{\i 2 \pi \phi(t)}$. Then we have that
\begin{align}\label{equ:redefine_f}
    f_{T}(x)=\int_{-T}^T\|\hat{f}_{T}\|_{1}e^{2\pi i (t x+\phi(t))}\nu(t)dt.
\end{align}

Now given a quantum state $\rho$, a Hamiltonian $H$ and a parameter $t \in [-T,T]$, we define two random variables $\bfX_t$ and $\bfY_t$ as follows. Let $\bfb_{I}$ and $\bfb_{S^\dagger}$ be the measurement outcome of the circuit in Figure \ref{fig:hadamard_test} with $\tau=2\pi t$ and $W=I$ or $S^\dagger$ (where $S$ is the phase gate), respectively. Then we define $\bfX_t=(-1)^{\bfb_I}$ and $\bfY_t=(-1)^{\bfb_{S^\dagger}}$. As mentioned in the main text, we have that 
\begin{align}\label{eq:exp_xt_yt}
    \mathbb{E}\left[\bfX_t\right]=\textrm{Re}\left(\operatorname{tr}\left[\rho e^{-2\pi \i Ht}\right]\right),\quad \mathbb{E}\left[\bfY_t\right]=\textrm{Im}\left(\operatorname{tr}\left[\rho e^{-2\pi \i Ht}\right]\right).
\end{align}

Now given a point $x\in\R$, we define the random variable $\bfZ(x)$ as follows. Let $\bft$ be a random variable with probability density function $\nu$. Then we define
\begin{align}
    \bfZ(x)\defeq\|\hat{f}_{T}\|_{1}\cdot e^{2\pi \i (\bft x+\phi(\bft))}\cdot \lrb{\bfX_\bft+\i \bfY_\bft}.    
    \label{eq:def_zx}
\end{align}

It turns out that $\bfZ(x)$ is an unbiased estimator of the convolution $f_T \ast p$ at point $x$: 
\begin{lemma}\label{lem:expectation_value}
For the random variable $\bfZ(x)$ defined as Eq.~\eqref{eq:def_zx}, we have that
\begin{align}
\mathbb{E}[\bfZ(x)]=(f_{T}\ast p)(x), ~&~\forall x \in \R.
\end{align}
\end{lemma}
\begin{proof}
Let us first consider the conditional expectation $\mathbb{E}[\bfZ(x)|\bft=t]$ for some $t\in [-T,T]$. By Eq.~\eqref{eq:exp_xt_yt} and the definition of $\bfZ(x)$ in Eq.~\eqref{eq:def_zx}, we get
\begin{align}
    \mathbb{E}[\bfZ(x)|\bft=t]=&~ \mathbb{E}\left[\|\hat{f}_{T}\|_{1}e^{2\pi \i (\bft x+\phi(\bft))}(\bfX_\bft + \i \bfY_\bft) \big| \bft=t\right] \nonumber\\
    =&~ \|\hat{f}_{T}\|_{1}e^{2\pi \i (t_0 x+\phi(t))}\operatorname{tr}\left[\rho e^{-2\pi \i Ht}\right].
\end{align}
By the law of total expectation, we have  
\begin{align}\label{equ:expect_z_compt}
    \mathbb{E}[\bfZ(x)]= &~ \int_{-T}^{T} \mathbb{E}[\bfZ(x)|\bft=t]\cdot \P{\bft=t}dt\notag\\
    = &~\int_{-T}^{T}\|\hat{f}_{T}\|_{1}e^{2\pi \i (t x+\phi(t))}\operatorname{tr}\left[\rho e^{-2\pi \i Ht}\right]\nu(t)dt\notag\\
    =&~ \int_{-T}^{T}\hat{f}_{T}(t)e^{2\pi \i t x}\operatorname{tr}\left[\rho e^{-2\pi \i Ht}\right]{\rm d}t,
\end{align}
where the last step follows from the definition of $\nu$ in Eq.~\eqref{equ:definition_nu} and the definition of $\phi(t)$.

It remains to prove that the above expression indeed coincides with $f_T\ast p(x)$. Indeed, we have that:
\begin{align}\label{equ:before_conv_p}
    (f_T\ast p)(x)= &~ \int_{-\infty}^\infty p(x-y)f_{T}(y){\rm d}y\notag\\
    = &~ \int_{-\infty}^\infty\int_{-T}^Tp(x-y)\hat{f}_{T}(t)e^{2\pi \i t y}{\rm d}t {\rm d}y\notag\\
    =&~\int_{-T}^T\hat{f}_{T}(t) {\rm d}t\int_{-\infty}^\infty p(x-y)e^{2\pi \i t y}{\rm d}y.
\end{align}
By the definition of $p(x)$ in Eq.~\eqref{eq:def_p}, we have that 
\begin{align}
    \int_{-\infty}^\infty p(x-y)e^{2\pi \i t y}{\rm d}y=\int_{-\infty}^\infty \sum_{k\geq 0}p_k\delta(x-y-E_k)e^{2\pi \i ty}{\rm d}y = \sum_{k\geq 0}p_ke^{2\pi \i t (x-E_k)},
\end{align}
where the last step follows from the integration of Dirac delta function. Then, it implies that
\begin{align}\label{eq:f_ast_p}
    (f_T\ast p)(x) = \int_{-T}^T\hat{f}_{T}(t) {\rm d}t \cdot \sum_{k\geq 0}p_ke^{2\pi \i t (x-E_k)}
    = \int_{-T}^{T}\hat{f}_{T}(t)e^{2\pi \i t x}\operatorname{tr}\left[\rho e^{-2\pi \i Ht}\right]{\rm d}t,
\end{align}
where the last step follows from $\operatorname{tr}\left[\rho e^{-2\pi \i H t}\right]=\sum_{k\geq 0} p_ke^{-2\pi \i tE_k}$.

Comparing Eqs.~\eqref{equ:expect_z_compt} and~\eqref{eq:f_ast_p}, we conclude that $\E{\bfZ(x)}=(f_T\ast p)(x)$ for all $x\in \R$. The lemma is thus proved.

\end{proof}

With Lemma~\ref{lem:expectation_value} established, it is now straightforward to analyze how many samples we need to estimate the function $f_{T}\ast p$ at various points within a target accuracy. 
\begin{lemma}[Sample complexity of the convolution evaluation]\label{lem:sampling_overhead}
Let $\{(t^{(i)},X^{(i)},Y^{(i)})\}_{i=1}^S$ be $S$ i.i.d. samples such that $t^{(i)}\sim \nu$, $X^{(i)} \sim \bfX_{t_i}$ and $Y^{(i)} \sim \bfY_{t_i}$, where $\nu$ is defined as Eq.~\eqref{equ:definition_nu}, and $\bfX_t$ and $\bfY_t$ are the measurement outcome of the circuit in Figure \ref{fig:hadamard_test} with $\tau=2\pi t$ and $W=I$ or $S$, respectively. Let $x_1,x_2,\ldots,x_M\in \R$ be arbitrary. For each $j\in [M]$, let $\overline{Z}_{j}$ be defined as follows:
\begin{align}
     \overline{Z}_j := \frac{\|\hat{f}_{T}\|_{1}}{S}\sum_{i=1}^S e^{2\pi \i (t^{(i)} x_j+\phi(t^{(i)}))}\cdot (X^{(i)}+\i Y^{(i)}).
     \label{eq:def_zj}
\end{align}
Then for any $\epsilon_1>0$ and $\delta_1 \in (0,1)$, letting
\begin{align}\label{equ:number_samples}
S \defeq \left\lceil \frac{\|\hat{f}_T\|_1^2 \ln{4M/\delta_1}}{\epsilon_1^2} \right\rceil,
\end{align}
we have
\begin{align}
\P{\forall j\in [M]: |\overline{Z}_j-(f_T\ast p)(x_j)|\leq \epsilon_1} \geq 1-\delta_1.
\end{align}
\label{prop:number_samples_bound}
\end{lemma}
\begin{proof}
Recall that $\bfZ(x)=\|\hat{f}_{T}\|_{1}\cdot e^{2\pi \i (\bft x+\phi(\bft))}\cdot \lrb{\bfX_\bft+\i \bfY_\bft}$ for any $x \in \R$. Then $\bar{Z}_j$ is the empirical mean of $S$ i.i.d. samples of $\bfZ(x_j)$ that correspond to $\{(t^{(i)},X^{(i)},Y^{(i)})\}_{i=1}^S$, for each $j \in [M]$. Since $\bfX_\bft$ and $\bfY_\bft$ take values in $\lrcb{1, -1}$, we know that $\rep\lrb{\bfZ(x)}$ and $\imp\lrb{\bfZ(x)}$ take values in $[-\|\hat{f}_{T}\|_{1}, \|\hat{f}_{T}\|_{1}]$.
It then follows from Hoeffding's inequality~\cite{Hoeffding1963} that for our choice of $S$ in Eq.~\eqref{equ:number_samples}, for any $j\in [M]$, it holds that
\begin{align}
    \P{|\rep\lrb{\overline{Z}_j}-\mathbb{E}[\rep\lrb{\bfZ(x_j)}]|> \frac{\epsilon_1}{\sqrt{2}}}<\frac{\delta_1}{2M}, \\
    \P{|\imp\lrb{\overline{Z}_j}-\mathbb{E}[\imp\lrb{\bfZ(x_j)}]|> \frac{\epsilon_1}{\sqrt{2}}}<\frac{\delta_1}{2M}.
\end{align}
Then by the triangle inequality and union bound, we get
\begin{align}
    \P{|\overline{Z}_j-\mathbb{E}[\bfZ(x_j)]|> {\epsilon_1}}<\frac{\delta_1}{M}.
\end{align}
Meanwhile, by Lemma~\ref{lem:expectation_value}, we know that
\begin{align}
    \mathbb{E}[\bfZ(x_j)] = (f_T\ast p)(x_j). 
\end{align}
Thus, we have
\begin{align}
    \P{|\overline{Z}_j-(f_T\ast p)(x_j)]|> {\epsilon_1} }<\frac{\delta_1}{M}.
\end{align}
By a union bound over all $j\in [M]$, we get that
\begin{align}
    \P{\exists j\in [M]:|\overline{Z}_j-(f_T\ast p)(x_j)]|> {\epsilon_1}} < \delta_1.
\end{align}
\end{proof}

\begin{remark}
Note that $\overline{Z}_j$ is a complex number in general, but $(f_T \ast p)(x_j)$ is real provided that $f_T$ is real. In this case, we can re-define $\overline{Z}_j$ as the real part of the right-hand side of Eq.~\eqref{eq:def_zj} and Lemma \ref{lem:sampling_overhead} will still hold. We envision that in some scenarios, it is useful to have a complex filter $f_T$, and hence define $\overline{Z}_j$ as Eq.~\eqref{eq:def_zj} for the sake of generality.    
\end{remark}

We have given a data structure $\textsc{ConvEval}$ in Algorithm \ref{alg:eval_conv} for evaluating the convolution $f_T \ast p$ at multiple points. Lemma \ref{lem:sampling_overhead} immediately implies that:
\begin{corollary}
Let $x_1,x_2,\ldots,x_M\in \R$ be arbitrary. Suppose the data structure \textsc{ConvEval} is initialized with parameters $(f_T, \epsilon, \delta, M)$. Let $h_j$ be the output of the procedure $\textsc{ConvEval.Eval}(x_j)$ for $j \in [M]$. Then we have 
\begin{align}
\P{\forall j \in [M]:|(f_T*p)(x_j)-h_j| \le \epsilon} \ge 1-\delta.
\end{align}
\label{cor:validity_conv_eval}
\end{corollary}
 
\begin{lemma}[Running time of the convolution evaluation data structure]
Suppose the data structure \textsc{FilterSampler} runs in ${\cal O}(1)$-time. Then, the data structure \textsc{ConvEval} in Algorithm~\ref{alg:eval_conv} has the following running times:
\begin{itemize}
    \item Procedure \textsc{Init}$(f_T, \epsilon,\delta, M)$ has ${\cal O}(T)$ maximal evolution time and ${\cal O}(ST)$ total evolution time, where $S={\cal O}(\epsilon^{-2} \|\hat{f}_T\|_1^2\log{\delta^{-1}M})$. 
    \item Procedure \textsc{Eval}$(x)$ has ${\cal O}(S)$ classical post-processing time.
\end{itemize}
\label{lem:conv_eval_cost}
\end{lemma}
\begin{proof}
The \textsc{ConvEval.Init} procedure runs the Hadamard test circuit $2S$ times to get the samples $\lrcb{(x^{(i)}, y^{(i)})}_{i=1}^S$. Since the filter function $f_T$ has spectrum bounded in $[-T, T]$, the maximal evolution time is $2 \pi T$ and the total evolution time is at most $4 \pi ST$.

The \textsc{ConvEval.Eval} procedure then uses the $S$ samples to compute the estimate of $(f_T\ast p)(x)$. Moreover, the computation is classical and elementary. 
\end{proof}

\subsection{Application to Gaussian derivative filters}
\label{subsec:cost_conv_eval_gdf}
In this appendix, we apply the data structure \textsc{ConvEval} to the band-limited Gaussain derivative filter $g_{\sigma, T}$:
\begin{align}
    g_{\sigma, T}(x) =  \int_{-T}^T \hat{g}_\sigma(\xi) e^{2\pi \i x \xi} d\xi=2 \pi \i \int_{-T}^T \xi e^{-\frac{1}{2} (\sigma \pi \xi )^2+2\pi \i x \xi} d\xi.
    \label{eq:def_g_sigma_t2}
\end{align}

To apply Lemma~\ref{lem:sampling_overhead}, we first bound the $L_1$-norm of its spectrum.
\begin{claim}\label{clm:gaussian_derivative_l1}
Let $g_{\sigma, T}$ be defined as Eq.~\eqref{eq:def_g_sigma_t2}. Then we have $\|\hat{g}_{\sigma,T}\|_1\leq \frac{4}{\pi\sigma^2}$.
\end{claim}
\begin{proof}
By the fact that $\hat{g}_{\sigma, T}(\xi) = \hat{g}_\sigma(\xi) {\bf 1}_{|\xi|\le T}$ and direct calculation, we obtain
\begin{align}
    \lonenorm{\hat{g}_{\sigma, T}}\le \lonenorm{\hat{g}_{\sigma}} 
    = \int_{-\infty}^{+\infty} \abs{2\pi \i \xi e^{-\frac{1}{2}(\sigma \pi \xi)^2}} d\xi = \int_{0}^{+\infty} 4\pi \xi e^{-\frac{1}{2}(\sigma \pi \xi)^2} d\xi
    = \frac{4}{\pi \sigma^2}.
\end{align}
\end{proof}

Then we get the following corollary on the sample complexity of evaluating $g_{\sigma,T}\ast p$ on $M$ points.

\begin{corollary}
\label{cor:gauss_deriv_conv_sample_complexity}
Let $\epsilon_1>0$, $\delta_1 \in (0,1)$ and $x_1,x_2,\dots,x_M \in \R$ be arbitrary. Suppose the data structure \textsc{ConvEval} is initialized with parameters $(g_{\sigma, T}, \epsilon_1, \delta_1, M)$. Let $h_j$ be the output of the procedure $\textsc{ConvEval.Eval}(x_j)$ for $j \in [M]$. Then we have 
\begin{align}
\P{\forall j \in [M]:|(g_{\sigma, T}*p)(x_j)-h_j| \le \epsilon_1} \ge 1-\delta_1.
\label{eq:validity_conv_eval_gaussian_derivative}
\end{align}
Furthermore, it take $S={\cal O}(\epsilon_1^{-2}\sigma^{-4}\log{M/\delta_1})$ samples from Hadamard tests to obtain $h_1$, $h_2$, $\dots,$ $h_M$.
\end{corollary}
\begin{proof}
Claim \ref{clm:gaussian_derivative_l1} implies $\|\hat{g}_{\sigma,T}\|_1={\cal O}(\sigma^{-2})$. Thus the procedure \textsc{ConvEval.Init}($g_{\sigma, T}, \epsilon_1, \delta_1, M$) draws $S={\cal O}(\epsilon_1^{-2} \|\hat{g}_{\sigma,T}\|_1^2 \log{M/\delta_1})
={\cal O}(\epsilon_1^{-2}\sigma^{-4}\log{M/\delta_1})$ samples from Hadamard tests. Then Eq.~\eqref{eq:validity_conv_eval_gaussian_derivative} follows immediately from Corollary \ref{cor:validity_conv_eval}.
\end{proof}

\section{Main Theorem}
\label{sec:formal_description}

In this appendix, we describe our main results about ground state energy estimation. {Recall that we have presented} a $\tilde{O}(\Delta^{-1})$-depth algorithm for GSEE in Algorithm \ref{alg:low_depth_gsee}.

\begin{theorem}[Ground state energy estimation]\label{thm:gsee_main_formal}
Let $H=\sum_{j=0}^{N-1} E_j \ketbra{E_j}{E_j}$ be a Hamiltonian such that $E_0 < E_1 \le E_2 \le \dots \le E_{N-1}$ are the eigenvalues of $H$, and the $\ket{E_j}$'s are orthonormal eigenstates of $H$. Suppose we are given access to the Hamiltonian evolution $e^{\i Ht}$ for any $t\in \R$. Let $\Delta>0$ be given such that $\Delta \leq E_1-E_0$. Moreover, suppose we can prepare a state $\rho$ such that $\bra{E_0} \rho \ket{E_0} \ge \eta$ for known $\eta > 0$.

{Let $\epsilon>0$ be small enough such that it satisfies the condition in Lemma \ref{lem:convolution_separation}, and let $\delta\in (0,1)$ be arbitrary. Then
the output of Algorithm \ref{alg:low_depth_gsee} (i.e. $\mathrm{GSEE}(H, \rho, \epsilon, \delta, \Delta, \eta)$) is $\epsilon$-close to $E_0$ with probability at least $1-\delta$. Furthermore, in this algorithm,}
\begin{itemize}
    \item The maximal Hamiltonian evolution time is $\tilde{{\cal O}}(\Delta^{-1})$;
    \item The total Hamiltonian evolution time is $\tilde{{\cal O}}(\eta^{-2}\epsilon^{-2}\Delta)$;
    \item The classical running time is $\tilde{\cal O}(\eta^{-2}\epsilon^{-3}\Delta^{3})$.
\end{itemize}
\end{theorem}
\begin{proof}

{We first prove the correctness of Algorithm \ref{alg:low_depth_gsee}.} By construction, $\tilde{E}_0$ satisfies Eq.~\eqref{eq:cond_wte0} in Lemma \ref{lem:gsee_correctnesss}. Meanwhile, by Claim~\ref{clm:l_infty_approx} and the choice of $T$ in Algorithm \ref{alg:gsee}, we have
\begin{align}
    \abs{(g_\sigma * p)(x) - (g_{\sigma, T} * p)(x)} \leq \frac{\tilde{\epsilon}}{2},  ~&~\forall x \in \R,
    \label{eq:dist_gsigmap_gsigmatp2}
\end{align}
Meanwhile, since $\textsc{ConvEval}$ is initialized with parameters $({H, \rho}, g_{\sigma, T}, \tilde{\epsilon}/2, \delta/2, M)$ in Algorithm \ref{alg:gsee}, by Corollary~\ref{cor:gauss_deriv_conv_sample_complexity}, we get
\begin{align}
\P{\forall j \in [M]:|(g_{\sigma, T}*p)(x_j)-h_j| \le \frac{\tilde{\epsilon}}{2}} \ge 1-\frac{\delta}{2}.
    \label{eq:prob_eval_err_large}
\end{align}
Then it follows from Eqs.~\eqref{eq:dist_gsigmap_gsigmatp2} and \eqref{eq:prob_eval_err_large} and the triangle inequality that
\begin{align}
\P{\forall j \in [M]:|(g_{\sigma}*p)(x_j)-h_j| \le \tilde{\epsilon}} \ge 1-\frac{\delta}{2},
\end{align}
which coincides with Eq.~\eqref{eq:cond_hjs} in Lemma \ref{lem:gsee_correctnesss}, given the choice of $\tilde{\epsilon}$ in Algorithm \ref{alg:gsee}. Now with both of its conditions met, Lemma \ref{lem:gsee_correctnesss} implies that the output of Algorithm \ref{alg:gsee}, i.e., $x_{j^*}$, is $\epsilon$-close to $E_0$ with probability at least $1-\delta$, as desired. 

{Next, we analyze the cost of Algorithm \ref{alg:low_depth_gsee}. In Line \ref{ln:get_coarse_estimate} of the algorithm,} we run the algorithm in \cite{lin2022heisenberg} to obtain $\tilde{E}_0$. Since $\sigma=\tilde{\Omega}(\Delta)$, this step has maximal evolution time $\tilde{\cal O}(\Delta^{-1})$, total evolution time $\tilde{\cal O}(\Delta^{-1}\eta^{-2})$, and classical post-processing time $\tilde{\cal O}(\Delta^{-1}\eta^{-2})$.

Then, in Line~\ref{ln:call_conv_init} of Algorithm~\ref{alg:gsee}, we run \textsc{ConvEval.Init}({$H$, $\rho$}, $g_{\sigma, T}$, $\tilde{\epsilon}$/2, $\delta/2$, $M$) to initialize the  data structure $\textsc{ConvEval}$ in Algorithm~\ref{alg:eval_conv}. We choose the parameters as follows: 
\begin{itemize}
    \item $\tilde{\epsilon}=\Omega(\epsilon \eta \sigma^{-3}) = \tilde{\Omega}(\epsilon \eta \Delta^{-3})$,
    \item $T=\tilde{\cal O}(\sigma^{-1})=\tilde{\cal O}(\Delta^{-1})$,
    \item $M=\Theta(\sigma \epsilon^{-1}) = \tilde{\Theta}(\Delta  \epsilon^{-1})$.
\end{itemize}
Thus, by Corollary~\ref{cor:gauss_deriv_conv_sample_complexity}, we have 
\begin{align}
    S=\tilde{\Theta}(\epsilon_1^{-2} \sigma^{-4}) = \tilde{\Theta}(\epsilon^{-2}\eta^{-2} \Delta^{6} \cdot \Delta^{-4}) = \tilde{\Theta}(\epsilon^{-2} \eta^{-2} \Delta^{2}).
\end{align}
The for-loop in Procedure \textsc{ConvEval.Init} of Algorithm~\ref{alg:eval_conv} draws $S$ samples from the Hadamard test circuit. The sampling process has the maximal evolution time $2\pi T=\tilde{\cal O}(\Delta^{-1})$ and total evolution time at most ${\cal O}(TS) = \tilde{\cal O}(\epsilon^{-2} \eta^{-2} \Delta)$.

Next, in Line~\ref{ln:forloop_gsee} of Algorithm~\ref{alg:gsee}, we call the procedure \textsc{ConvEval.Eval} $M$ times to evaluate the convolutions at $x_1,x_2,\dots,x_M$. Each evaluation takes ${\cal O}(S)=\tilde{\cal O}(\epsilon^{-2} \eta^{-2} \Delta^{2})$ classical time. Hence, this step takes ${\cal O}(MS) = \tilde{\cal O}(\Delta \epsilon^{-1}\cdot \epsilon^{-2} \eta^{-2} \Delta^{2})=\tilde{\cal O}(\epsilon^{-3}\eta^{-2}\Delta^3)$ classical time.

Combining these steps together, we get that the whole GSEE algorithm takes:
\begin{itemize}
    \item maximal evolution time $\tilde{\cal O}(\Delta^{-1})$,
    \item total evolution time $\tilde{\cal O}(\Delta^{-1}\eta^{-2}+\epsilon^{-2}\eta^{-2}\Delta) = \tilde{\cal O}(\epsilon^{-2}\eta^{-2}\Delta)$, and 
    \item classical post-processing time $\tilde{\cal O}(\epsilon^{-3}\eta^{-2}\Delta^3)$,
\end{itemize}
as claimed.
\end{proof}

{As described in the introduction, it is favorable to be able to reduce the maximal evolution time (or circuit depth) per circuit run at the cost of a larger total evolution time. If we accept that the maximal evolution time is a proxy for the number of gates required to implement a time evolution on a fault-tolerant device, a smaller maximal evolution time implies a smaller number of gates, which in turn means that the circuit can be run reliably with a higher noise rate per gate. Thus, running the circuit with a smaller code distance is possible, which translates to a smaller number of physical qubits, bringing this application closer to reality. However, note that, for the scope of this paper, we do not envision running this algorithm on NISQ devices, as we do not analyze the effect of noise on its performance.}

This allows one to make the most use of the available circuit depth afforded by the quantum architecture.
Such a feature is desirable in the era of early fault-tolerant quantum computing where there is likely to be a limit to the available coherence of the device \cite{tong2022designing}.
Fortunately, this feature follows directly from the above theorem and we present it as a corollary below.
Note that in Theorem \ref{thm:gsee_main_formal}, $\Delta$ is merely a lower bound on the true spectral gap $\Delta_{{\rm true}} \defeq E_1 - E_0$ of Hamiltonian $H$, not necessarily $\Delta_{{\rm true}}$ itself. In fact, $\Delta$ can range from $\tilde{\cal O}(\epsilon)$ (in order to satisfy the condition in Lemma \ref{lem:convolution_separation}) to $\Delta_{{\rm true}}$. By setting $\Delta=\tilde{\cal O}(\epsilon^{\alpha} \Delta_{{\rm true}}^{1-\alpha})$ with $\alpha \in [0, 1]$, we obtain:
\begin{corollary}
\label{rem:interpolation}
Let $H=\sum_{j=0}^{N-1} E_j \ketbra{E_j}{E_j}$ be a Hamiltonian such that $E_0 < E_1 \le E_2 \le \dots \le E_{N-1}$ are the eigenvalues of $H$, and the $\ket{E_j}$'s are orthonormal eigenstates of $H$. Let $\Delta_{\rm true}=E_1-E_0$ be the spectral gap of $H$. Suppose we are given access to the Hamiltonian evolution $e^{\i Ht}$ for any $t\in \R$. Moreover, suppose we can prepare a state $\rho$ such that $\bra{E_0} \rho \ket{E_0} \ge \eta$ for known $\eta > 0$.

Then for any $\alpha \in [0, 1]$, for sufficiently small $\epsilon>0$ and any $\delta\in (0,1)$, there exists an algorithm that estimates $E_0$ within accuracy $\epsilon$ with probability at least $1-\delta$ such that:
\begin{itemize}
    \item The maximal Hamiltonian evolution time is $\tilde{{\cal O}}(\epsilon^{-\alpha} \Delta_{\rm true}^{-1+\alpha})$;
    \item The total Hamiltonian evolution time is $\tilde{{\cal O}}(\eta^{-2}\epsilon^{-2+\alpha}\Delta_{\rm true}^{1-\alpha})$;
    \item The classical running time is $\tilde{\cal O}(\eta^{-2}\epsilon^{-3+\alpha}\Delta_{\rm true}^{3-\alpha})$.
\end{itemize}
\end{corollary}
In particular, setting $\alpha=0$ or $1$ leads to:
\begin{itemize}
    \item $\Delta = \Delta_{{\rm true}}$, for which Theorem \ref{thm:gsee_main_formal} yields an algorithm with maximal evolution time $\tilde{\cal O}(\Delta_{{\rm true}}^{-1})$ and total evolution time $\tilde{\cal O}(\eta^{-2} \epsilon^{-2} \Delta_{{\rm true}})$; or  
    \item $\Delta = \tilde{\cal O}(\epsilon)$, for which Theorem \ref{thm:gsee_main_formal} yields an algorithm with maximal evolution time $\tilde{\cal O}(\epsilon^{-1})$ and total evolution time $\tilde{\cal O}(\eta^{-2} \epsilon^{-1})$ (i.e., the Heisenberg limit).
\end{itemize}
For general $\Delta = \tilde{\cal O}(\epsilon^{\alpha} \Delta_{{\rm true}}^{1-\alpha})$ with $\alpha \in [0, 1]$, Theorem \ref{thm:gsee_main_formal} yields an algorithm with maximal evolution time $\tilde{\cal O}(\epsilon^{-\alpha}\Delta_{{\rm true}}^{-1+\alpha})$ and total evolution time $\tilde{\cal O}(\eta^{-2} \epsilon^{-2+\alpha} \Delta_{{\rm true}}^{1-\alpha})$. In other words, tuning $\Delta$ between the two extremes gives a trade-off between the circuit depth and total runtime of the algorithm.

\section{Comparison to the approach of Lin et al}\label{sec:comparison_lin}
The main advantage of our approach compared to~\cite{lin2022heisenberg} is in the minimal evolution time required to achieve a desired precision. Indeed, in their approach the evolution time scales inverse linearly with the desired precision. For our approach, the minimal evolution time is dictated by the reciprocal of the energy gap of the Hamiltonian and any additional precision we wish to attain only causes a poly-logarithmic factor in the evolution time. 
Of course, this comes at the expense of a higher sample complexity at smaller evolution times. This trade-off between the evolution time and the sample complexity is discussed in Corollary~\ref{rem:interpolation}.

Note that this improvement in the minimal evolution time comes from two conceptual differences in our approach compared to~\cite{lin2022heisenberg}: the choice of the filter function (Heaviside versus Gaussian derivative) and how we then infer the value of the ground state energy from the convolution (jump versus $0$ of derivative).

Both our approach and that of \cite{lin2022heisenberg} require a truncated approximation of the underlying filter function to implement the algorithm with only finite-time evolutions. However, as the Heaviside function has a jump at $0$, the degree of the Fourier series necessarily has to increase the better we want the approximation outside a small neighborhood of $0$ to be. For instance, in~\cite{lin2022heisenberg} they find an approximation
$F_{d,\epsilon}$ such that for $d=\mathcal{O}(\epsilon^{-1}\log{\epsilon^{-1}\delta^{-1}})$ and 
\begin{align}\label{equ:fourier_series_delta}
    F_{d,\epsilon}(x)=\frac{1}{\sqrt{2\pi}}\sum\limits_{k=-d}^d\hat{F}_{d,\epsilon,k}e^{ikx},
\end{align}
we have\footnote{Note that \cite{lin2022heisenberg} adopts different notations. There $\delta$ is the precision with which we approximate the ground state energy, and here it is $\epsilon$.} 
\begin{enumerate}
\item $-\frac{\delta}{2}\leq F_{d,\epsilon}(x)\leq 1+\delta $ for all $x\in\mathbb{R}$.
\item $|F_{d,\epsilon}(x)-\Theta(x)|\leq \delta$ for all $x\in[-\pi+\epsilon,-\epsilon]\cup[\epsilon,\pi-\epsilon]$.
\end{enumerate}
Note that the evolution time required to implement the representation in Eq.~\eqref{equ:fourier_series_delta} is $\mathcal{O}(d)$. This scales logarithmically with the precision with which we approximate the Heaviside function outside of the intervals $[-\pi+\epsilon,-\epsilon]\cup[\epsilon,\pi-\epsilon]$ and inverse-polynomially in size of the interval around $0$, $(-\epsilon,\epsilon)$, where we are not guaranteed that the two functions are close.

The approach of~\cite{lin2022heisenberg} consists of finding the smallest point $x$ such that $(p\ast F_{d,\epsilon})(x)\geq \eta$. If we were convolving with the Heaviside function, this would correspond to the ground state energy. But we will now argue that the neighborhood around $0$ for which we have the approximation can shift where the jump occurs. Indeed, note that for $x\in[E_0-\epsilon,E_0+\epsilon]$ and $\epsilon\leq \Delta$ we have that:
\begin{align}\label{equ:convolution_Heaviside_approx}
(p\ast F_{d,\epsilon})(x)&=(p\ast \Theta)(x)+\int_{-\epsilon}^{\epsilon} p(x-y)(F_{d,\epsilon}(y)-\Theta(y))dy+\int\limits_{[-1,1]\backslash[-\epsilon,\epsilon]} p(x-y)(F_{d,\epsilon}(y)-\Theta(y))dy\nonumber
\end{align}
Note that as $|F_{d,\epsilon}(x)-\Theta(x)|\leq \delta$ for all $x\in[-\pi+\epsilon,-\epsilon]\cup[\epsilon,\pi-\epsilon]$ we have that
\begin{align}
    \left|\int\limits_{[-1,1]\backslash[-\epsilon,\epsilon]} p(x-y)(F_{d,\epsilon}(y)-\Theta(y))dy\right|\leq \delta.
\end{align}
However, as we only have the promise that $-\frac{\delta}{2}\leq F_{d,\epsilon}(x)\leq 1+\delta $ for points in $[-\epsilon,\epsilon]$, we will not be able to infer the precise point of the jump at a precision larger than $\mathcal{O}(\epsilon)$ with this approach. This is because the residual integral term (i.e. the one over $(-\epsilon,\epsilon)$ in Eq.~\eqref{equ:convolution_Heaviside_approx}) will cause fluctuations in this interval and we will not be able to pin down the jump. 

On the other hand, by choosing our filter to be Gaussian derivatives, we are able to obtain a good approximation everywhere on the real line. Furthermore, by choosing the zeros of the derivative as criteria, we only need to make sure that the standard deviation is small enough to separate different eigenvalues. This way we obtain a smaller maximal evolution time.

\end{document}